\documentclass[letterpaper,twocolumn,10pt]{article}
\usepackage{usenix}

\usepackage{tikz}
\usepackage{amsmath}

\usepackage{filecontents}
\usepackage{microtype}
\usepackage{pifont}
\usepackage{tcolorbox}
\usepackage{multicol}
\usepackage{multirow}

\usepackage{bbm}
\usepackage{comment}
\usepackage{graphicx}
\usepackage{subfigure}
\usepackage{url}
\usepackage{amsmath}
\usepackage{caption}
\usepackage{mathtools}
\usepackage{enumitem}
\usepackage{booktabs}
\usepackage{amsmath}
\usepackage{tabularray}
\usepackage{tabularx}
\usepackage{tablefootnote}
\UseTblrLibrary{booktabs}

\usepackage{algorithmic}
\usepackage{float}
\usepackage{algorithm}
\usepackage{balance}
\usepackage{bbding}
\usepackage{pifont}
\usepackage{framed}
\usepackage{epigraph}
\usepackage{csquotes}
\usepackage[T1]{fontenc}
\usepackage{tablefootnote}
\usepackage{soul}
\usepackage{amsfonts}
\usepackage[dvipsnames]{xcolor}
\usepackage{booktabs}
\usepackage{multicol}
\usepackage{multirow}
\usepackage{threeparttable}
\usepackage[subtle]{savetrees}
\usepackage{wasysym}

\newcommand{\kw}[1]{\texttt{#1}}

\newcommand{\inlinedsection}[2][3pt]{\vspace{#1}\noindent\textbf{#2}.}

\newcommand{\inlinedsectionbf}[2][3pt]{\vspace{#1}\noindent\textbf{#2}:} 
\newcommand{\inlinedsectionbff}[2][3pt]{\vspace{#1}\noindent\textbf{#2}} 
\newcommand{\proj}{{\textsc{Chameleon\-Disc}}}

\usepackage{tikz}

\usepackage{enumitem}
\usepackage{amsthm}
\newtheorem{proposition}{Proposition}

\begin{document}

\setlength{\abovedisplayskip}{3pt}
\setlength{\belowdisplayskip}{3pt}
\setlength{\abovedisplayshortskip}{2pt}
\setlength{\belowdisplayshortskip}{2pt}
%-------------------------------------------------------------------------------

\title{\Large \bf Rethinking Software-Defined Networking Link Discovery \\ with Dynamic Randomization}

%for single author (just remove % characters)
\author{
{\rm Mingming Chen}\\
University of Texas at Dallas
\and
{\rm Teryl Taylor}\\
IBM Research
% copy the following lines to add more authors
\and
{\rm Frederico Araujo}\\
IBM Research
\and
{\rm Benjamin E. Ujcich}\\
Georgetown University
\and
{\rm Thomas La Porta}\\
The Pennsylvania State University
\and
{\rm Trent Jaeger}\\
University of California, Riverside
} % end author

\maketitle

%-------------------------------------------------------------------------------
\begin{abstract}
Software-defined networking (SDN) separates the control and data planes, enabling programmable, centralized network management. A core SDN service is topology discovery --- a periodic process that identifies network links. However, our analysis of five open-source SDN controllers and ten discovery protocols shows that all remain vulnerable to at least one form of link-fabrication attack. The common root cause is their reliance on traditional Link Layer Discovery Protocol (LLDP) packets, whose static identifiers expose their discovery purpose and attract adversaries to exploit them.

We propose \proj{}, a dynamic link-discovery protocol that simultaneously prevents and detects topology-poisoning attacks by eliminating this root cause. Our key insight is that an SDN controller can infer topology without embedding meaningful information in discovery packets. 
\proj{} removes targetable static LLDP signatures and employs a moving-target defense that combines decoy, obfuscation, and camouflage techniques instantiated dynamically at runtime to prevent topology poisoning and detect manipulation. We implement \proj{} on OpenDaylight and demonstrate effectiveness against all identifier-based topology-poisoning attacks. On a 252-link topology, median convergence is 1.37--4.53\,s for legitimate link changes and 4.06\,s for malicious relay detection, with 13.0 percentage points of additional mean CPU and negligible retained-heap difference. Across topologies up to 816 links, \proj{} provides a tunable security--performance trade-off; expanding discovery intervals reduces CPU and mapping-state at the cost of increased attack-detection latency.
\end{abstract}

\section{Introduction}\label{sec:introduction}

Software-defined networking (SDN) decouples the control plane from the data plane, enabling centralized, programmable network management. This flexibility supports fine-grained control and global optimization across diverse environments, including network traffic sampling~\cite{chen2025efficient}, %Internet of Things~\cite{farris2018survey}, 
cloud computing~\cite{yan2015software}, 
data centers~\cite{wang2017efficient}, and vehicular networks~\cite{garg2019sdn}.

SDN centralization enables several services, including a network topology service that discovers the ground-truth layout of the network. Topology discovery is a critical function in both traditional networks and SDN, as it provides essential real-time information for efficient traffic routing. In traditional networks, this is typically accomplished using the standard Link Layer Discovery Protocol (LLDP)~\cite{lldp}.
However, there is no standardized topology discovery protocol for SDN. As a result, many SDN implementations repurpose LLDP for convenience~\cite{onos_disc,odl_openflowplugin}, often without accounting for the unique security implications of SDN. 
Most open-source SDN controllers adopt the de facto OpenFlow Discovery Protocol (OFDP), which uses LLDP packets---or, in the case of hybrid networks that include both SDN and traditional switches, combines them with the Broadcast Domain Discovery Protocol (BDDP)~\cite{onos_disc, floodlight}---to discover links. In addition, nine SDN discovery protocols derived from OFDP also depend on LLDP~\cite{chen2024manipulating}.

There has been a growing number of recent attacks targeting SDN discovery protocols across diverse attack surfaces. On the control plane, adversaries exploit real vulnerabilities---e.g., CVE-2024-46942~\cite{cve202446942} and CVE-2024-46943~\cite{cve202446943}---to inject malicious flow entries, enabling flow entry-induced, globalized topology poisoning attacks~\cite{chen2024manipulating} (e.g., CVE-2024-37018~\cite{cve202437018}). On the data plane, malicious hosts or switches exploit vulnerabilities~\cite{thimmaraju2018taking} such as CVE-2016-2074~\cite{cve20162074} and CVE-2016-10377~\cite{cve201610377} to intercept and relay LLDP packets, enabling localized topology poisoning attacks~\cite{hong2015poisoning,skowyra2018effective,alimohammadifar2018stealthy,marin2019depth, dhawan2015sphinx, jero2017identifier}. These CVEs demonstrate that such vulnerabilities are real and actively exploited, yet the continued reliance on LLDP leaves existing discovery mechanisms exposed.

Applying traditional LLDP to SDN is problematic because the discovery process in a traditional network relies on each switch forwarding source information to its neighbors using LLDP packets—a model designed for decentralized architectures. However, this assumption is unnecessary in SDN, where a centralized controller manages discovery and can utilize entirely different strategies. While LLDP is effective at discovering links in SDN, its standardized packet type and format unintentionally expose its link discovery purpose, making it vulnerable to link fabrication attacks (e.g., ~\cite{chen2024manipulating,hong2015poisoning,skowyra2018effective,alimohammadifar2018stealthy,nehra2019sldp}).

Designing a lightweight and secure SDN link discovery protocol to prevent a broad range of topology poisoning attacks remains a significant challenge. The decoupling of the control and data planes in SDN introduces vulnerabilities to information manipulation, exposing numerous attack vectors. Existing countermeasures~\cite{hong2015poisoning,skowyra2018effective,baidya2020link,smyth2017detecting,huang2020towards} typically rely on known attack signatures and detect attacks only after they occur---often with substantial monitoring overhead~\cite{dhawan2015sphinx,alimohammadifar2018stealthy}. Among the ten SDN link discovery protocols we studied, only 3 incorporate security enhancements~\cite{nehra2019sldp, nehra2019tilak,azzouni2018softdp}. However, these either fail to prevent relay-based and 
newly discovered 
flow entry-induced topology poisoning attacks~\cite{alimohammadifar2018stealthy, chen2024manipulating}, or they are not fully compatible with SDN, as they require modifications to SDN switches.
While cryptographic mechanisms can ensure message integrity, they do not prevent all forms of topology manipulation. Attacks such as relay-based or flow entry-induced topology poisoning do not alter packet content or fabricate discovery packets; instead, they exploit forwarding behavior. As a result, even with strong integrity guarantees, adversaries can still mislead the controller about the actual network topology.

To address these challenges, we introduce \proj{}, a lightweight SDN link discovery protocol fully compatible with existing SDN architectures that, for the first time, prevents and detects all known topology-poisoning attacks exploiting static identifiers or discovery signatures in discovery packets~\cite{hong2015poisoning,baidya2020link,nehra2019sldp,alimohammadifar2018stealthy,chen2024manipulating}.
Our key insight is that \emph{an SDN controller can infer topology without embedding meaningful information in discovery packets that can be abused by the attacker.} 
\proj{} provides robust protection against a wide range of threats originating from diverse attack surfaces via orchestrating a three-role discovery ensemble: (1) DecoyType packets (LLDP/BDDP) lure adversaries targeting traditional discovery traffic; (2) MorphType packets randomize all fields, obfuscating discovery intent and blocking attacks that exploit static identifiers; and (3) CamoType (ARP/IP) packets blend with normal traffic to validate genuine connectivity and expose malicious switches capable of advanced packet analysis and manipulation. Crucially, our camouflage layer creates an inescapable dilemma for attackers: an attacker who manipulates normal data-plane traffic exposes themselves through network disruption, while attackers who avoid attacking the data plane inadvertently render their topology manipulation detectable. By cross-checking link inferences from these heterogeneous probes through multiple validations, the controller distinguishes legitimate link changes from fabricated ones, attributes inconsistencies to the responsible host, application, or switch, and updates topology only after consecutive, consistent confirmations. 

This strategy prevents adversaries from using the packets’ static identifier or discovery signature to persistently fabricate false links, effectively mitigating multiple attack vectors. More importantly, the dynamic obfuscated discovery packets preserve correct SDN topology discovery: because the controller initiates the discovery packets and receives them after a one-link traversal, it can accurately recover the link source using a pre-stored mapping between the dynamically randomized packet and its true origin, and retrieve the link destination from metadata provided by the reporting switch.

Implementing \proj{} is challenging for several reasons. First, constructing the discovery packet and selecting its type is critical to ensuring one-link traversal. 
Achieving this is non-trivial across different packet types. For obfuscated MorphType packets and camouflage packets, we must ensure that they exclusively match the lowest-priority \textit{table-miss} entries on the corresponding switches, without being intercepted by higher-priority flow entries.
Second, using camouflage ARP packets for verification requires avoiding interference with controller services (e.g., ARP Handler and Host Tracker). \proj{} therefore constructs CamoType probes using fresh, non-conflicting MAC/IP identities that mimic legitimate new-host arrivals and adds filters to ARP-related network applications.
Third, generating randomized values for each field of the discovery packets in a time- and space-efficient manner requires careful consideration.  
Finally, realizing built-in detection without additional real-time monitoring is challenging, as it requires integrating prevention and detection within the lightweight discovery process.

Our approach recognizes that link discovery in SDN environments differs fundamentally from traditional network architectures.
\proj{} ensures compatibility with trending SDN southbound protocols (e.g., OpenFlow and P4-runtime), 
eliminating the need for additional monitoring systems or modifications to standard SDN switches.
By comparison, existing defenses fall short in addressing all attack vectors. They require significant architectural modifications and add significant complexity caused by real-time monitoring or logic trees, resulting in higher resource consumption and reduced system efficiency. By focusing on the protocol's inherent vulnerabilities, our approach overcomes these limitations. 

We summarize our contribution as follows:
\begin{itemize}[leftmargin=15pt, itemsep=0pt, parsep=0pt, topsep=2pt, partopsep=2pt]

    \item We propose \proj{}, a dynamically randomized link discovery protocol that prevents and detects all known topology-poisoning attacks exploiting LLDP-based static identifiers, while remaining lightweight and fully compatible with existing SDN architectures. 
    \item We present a theoretical analysis of \proj{}'s resilience under adversarial conditions. With three-round validation, MorphType alone yields an expected random-guessing attack time of approximately 44 years, while CamoType provides an additional verification layer against adversaries that evade randomized discovery. 
    \item We implement \proj{} on OpenDaylight and evaluate its security, convergence, and performance. \proj{} defeats all four identifier-based topology-poisoning attacks as well as an advanced relay attack. On a 252-link topology, median convergence is 1.37--4.53\,s for legitimate link changes and 4.06\,s for malicious relay detection. Across topologies up to 816 links, \proj{} maintains bounded mapping state, and its discovery interval provides an explicit trade-off between detection latency and controller overhead.
\end{itemize}

\section{Motivation and Background}\label{sec:overview}

\begin{figure*}[t]
\centering
\subfigure[OFDP discovers a real link $A2\rightarrow1B$]{\label{fig:link_discovery} 
\includegraphics[width=0.32\textwidth]
{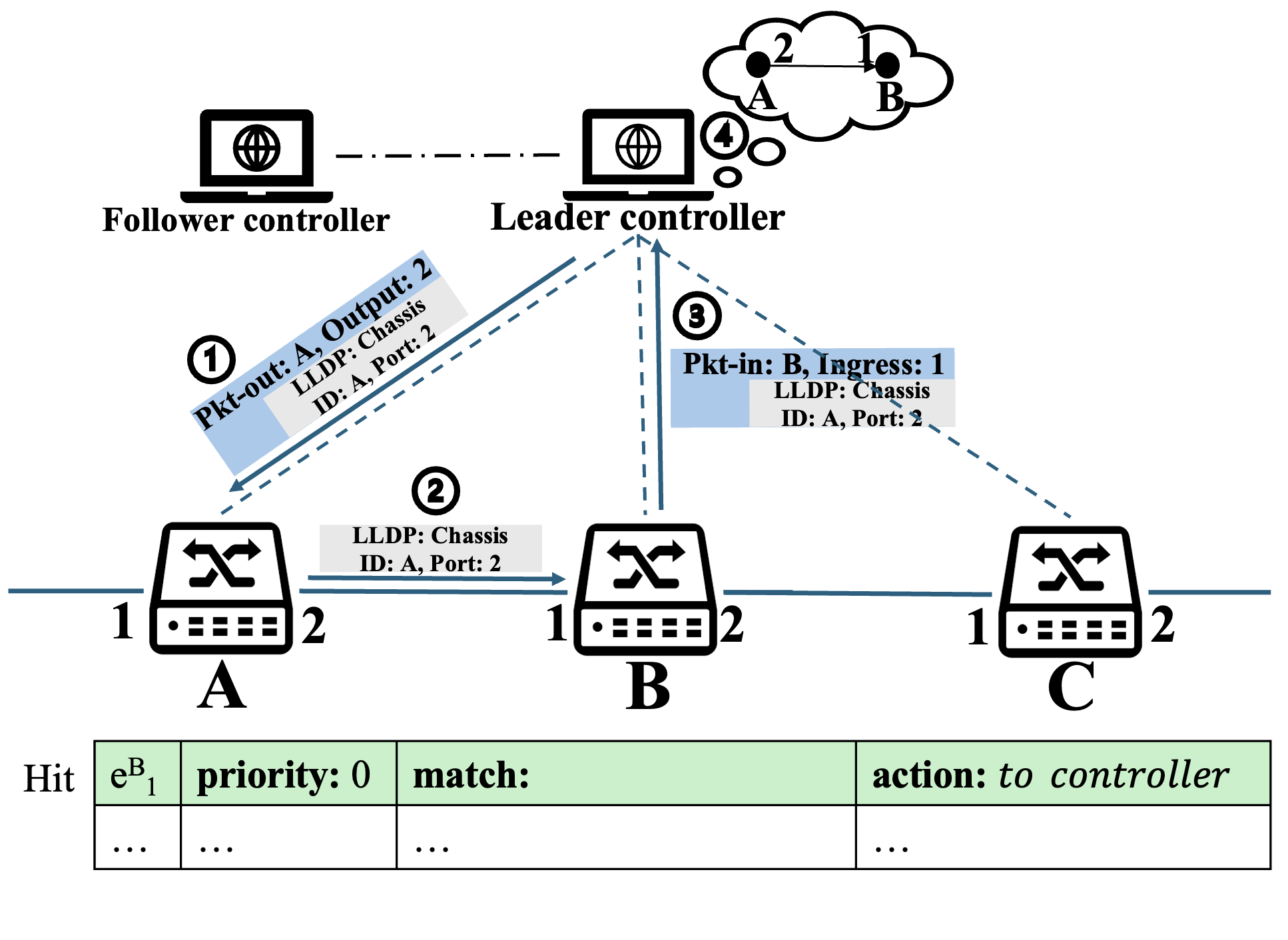}}
\hfill
\subfigure[Marionette\cite{chen2024manipulating} fabricates a link $A1\rightarrow1C$ ]{\label{fig:link_fab} \includegraphics[width=0.32\textwidth]{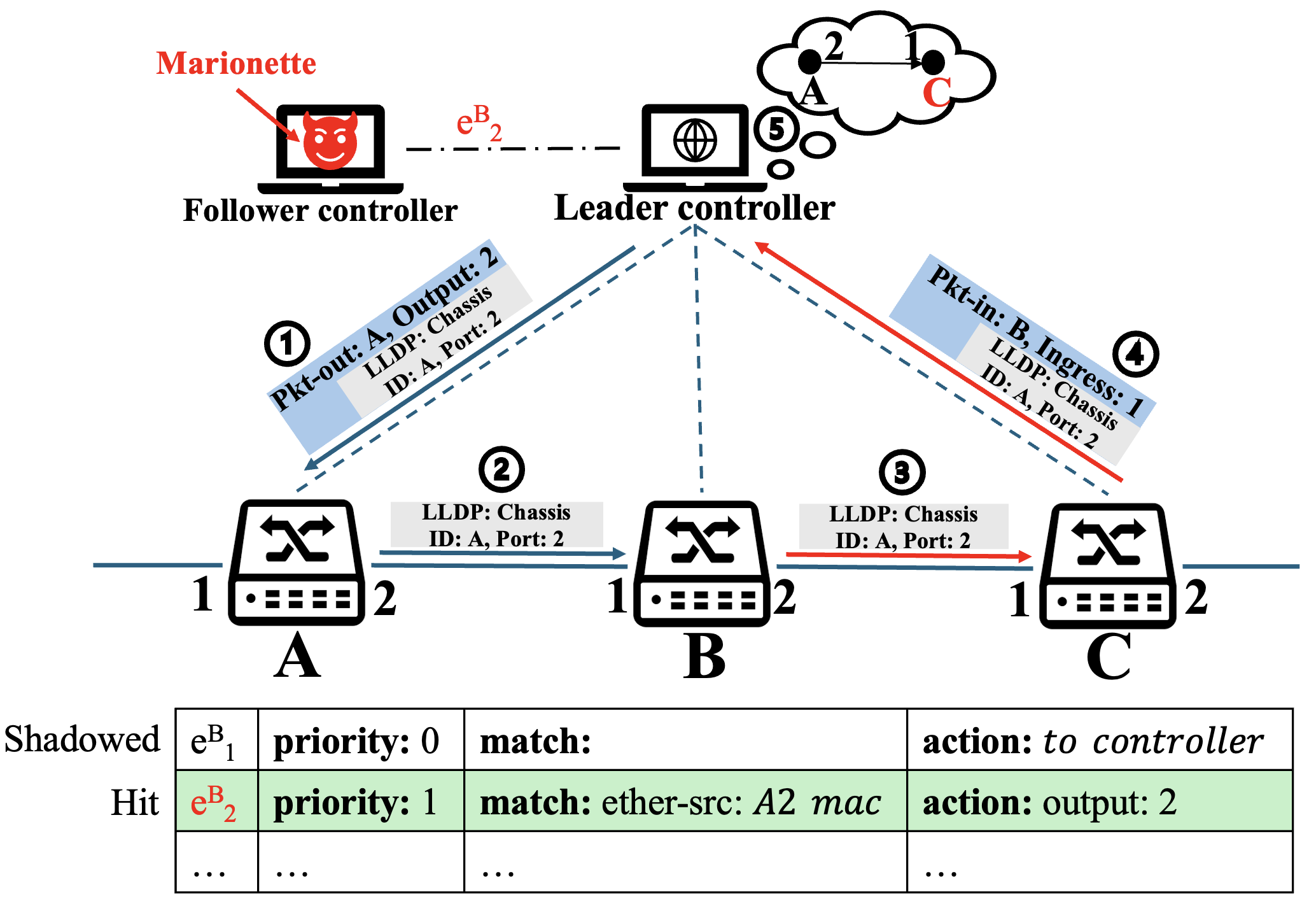}}
\hfill
\subfigure[\proj{} prevents Marionette]{\label{fig:morphDisc} \includegraphics[width=0.32\textwidth]{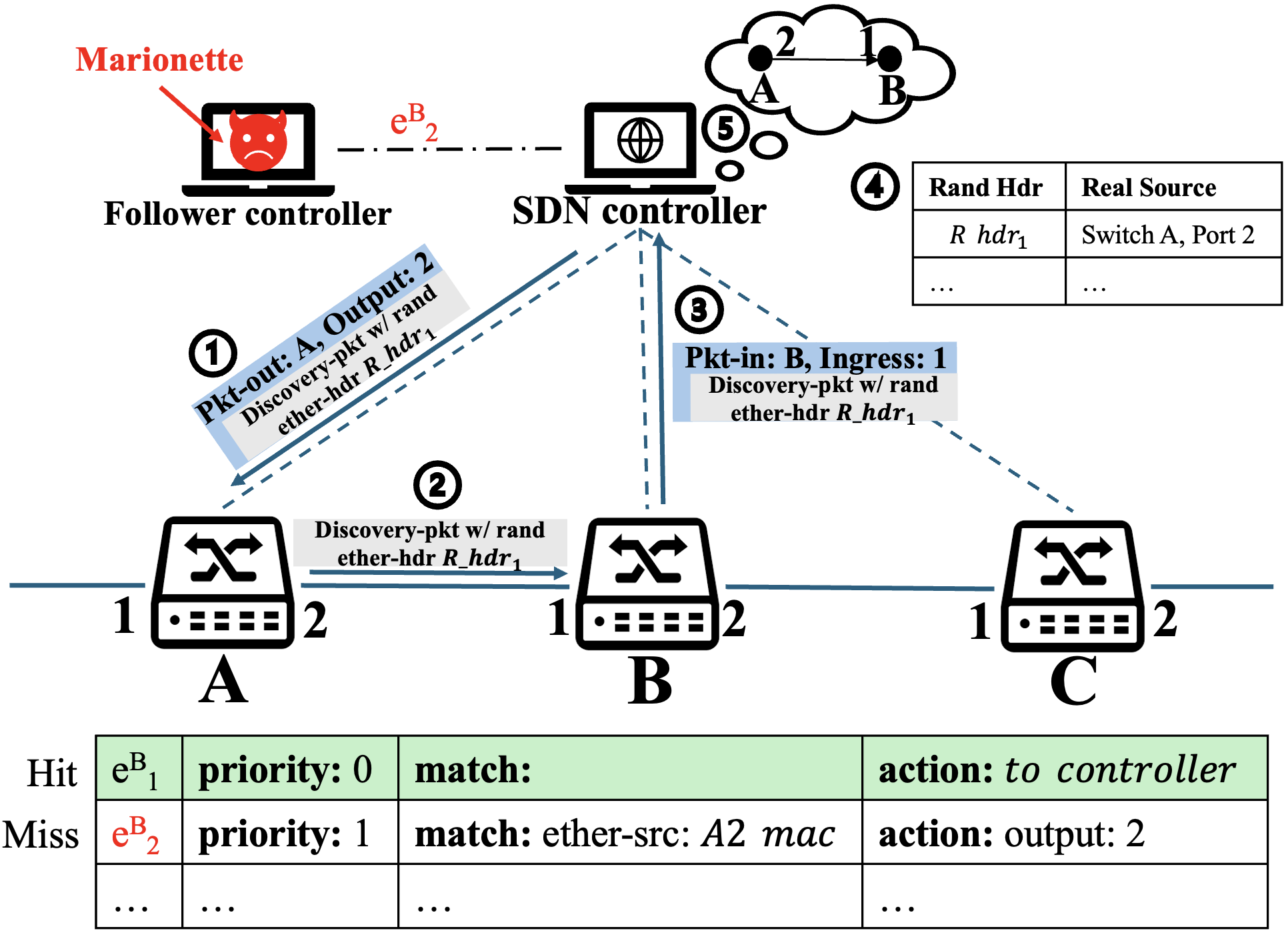}}
\vspace{-5pt}
\caption{Motivating Example for \proj{}}
\vspace{-5pt}
\label{fig:motivating_example}
\end{figure*}
%In an SDN architecture, 
OpenFlow is a standard SDN southbound protocol enabling a centralized controller to communicate directly with 
%OpenFlow 
switches. Once the controller establishes a connection with the switches, it gathers essential switch hardware information and data statistics. The central controller periodically discovers the network topology as the basis for instructing switches on how to forward traffic. The controller can \emph{proactively} install forwarding rules to all the connected switches---an approach known as proactive forwarding. SDN also supports a key feature called \emph{reactive} forwarding: if an arriving packet does not match any existing flow entry that can forward it normally in the data plane, the \textit{table-miss} flow entry sends this packet to the controller in a \kw{packet-in} message. 

The de-facto OpenFlow Discovery Protocol (OFDP) uses \kw{packet-out} and \kw{packet-in} messages to carry LLDP packets for link discovery. %Specifically, the controller constructs LLDP packets for each switch port and sends them via \kw{packet-out}. It also pre-installs corresponding flow entries (e.g., \textit{table-miss}) on all switches to ensure that, after traversing a link, the LLDP packet hits on that flow entry, encapsulated in a \kw{packet-in} message, and sent back to the controller. The controller extracts the link source information from the LLDP packet and parses the link destination information from the \kw{packet-in} message.
%A significant advantage of OFDP is that it requires no changes to OpenFlow switches; all logic resides in the SDN controller. 
%
However, because OFDP relies on static LLDP packets, it introduces vulnerabilities in the link discovery process, as programmable forwarding can be exploited.  %We demonstrate this problem in the example below, variance among OFDP implementations, and current defenses and their limitations.

\subsection{Motivating Example}\label{sec:motivation}

In Figure~\ref{fig:link_discovery}, we show how OFDP is used for link discovery normally.  A controller composes an LLDP packet for Port 2 on Switch A (denoted $LLDP_{A2}$) to detect the link originating from that port. In Step \ding{192}, the controller issues Switch A a \kw{packet-out} message containing $LLDP_{A2}$ with \kw{action} of instructing it to forward $LLDP_{A2}$ via Port 2. The $LLDP_{A2}$ packet has an \kw{ether-src} value that matches the MAC address of Port 2 on Switch A. In Step \ding{193}, the neighboring Switch B receives $LLDP_{A2}$ on its Port 1. In Step \ding{194}, finding no flow entry match, Switch B’s table-miss flow entry instructs the switch to forward the  $LLDP_{A2}$ packet back to the controller within \kw{packet-in} message. In Step \ding{195}, upon receiving this \kw{packet-in} message, the controller knows the communication originally emanated from Switch A’s Port 2 (by parsing $LLDP_{A2}$) and has now returned from Switch B’s Port 1 (by reading \kw{ingress} field of \kw{packet-in}), thus concluding a valid link $A2-1B$.

Figure~\ref{fig:link_fab} shows how a recent flow entry-induced poisoning attack, called Marionette~\cite{chen2024manipulating}, can subvert this normal discovery process to fabricate links. If Marionette resides in a follower controller or an SDN application, it can install a malicious, higher-priority flow entry $e_2^B$ on Switch B. After Step \ding{193}, because $LLDP_{A2}$ has $A2\_mac$ as \kw{ether-src}, it matches $e_2^B$ instead of the table-miss entry. As a result, Switch B forwards $LLDP_{A2}$ out of Port 2 instead of returning it to the controller at Step \ding{194}. Switch C, connected at that port, treats $LLDP_{A2}$ as a table-miss packet and sends it back to the controller as Step \ding{195}. At Step \ding{196}, the controller mistakenly believes the packet traveled from Switch A’s Port 2 to Switch C’s Port 1, thereby fabricating a non-existent link $A2-1C$.

The core of topology poisoning attacks lies in misleading the controller into inferring incorrect link source (\textit{link-src}) or link destination (\textit{link-dst}) information. Since the \textit{link-src} is extracted directly from the LLDP packet, fabricating it typically involves spoofing an LLDP packet with falsified source information. While such attacks can be mitigated by validating the integrity of LLDP packets,  fabricating \textit{link-dst} is more %complex and 
challenging to detect. The \textit{link-dst} is determined by identifying the switch that returns the LLDP packets to the controller. Any method that causes LLDP packets to traverse multiple links can also lead to an incorrect \textit{link-dst} inference. %\footnote{One variant involves a malicious switch forging ingress metadata in the \kw{packet-in} header~\cite{baidya2020link}, but this is inherently mitigated by our discovery packet obfuscation, like other \kw{link-dst} fabrication attacks.}. 
The example in Figure~\ref{fig:link_fab} demonstrates how a malicious application or backup controller can manipulate LLDP packet forwarding by injecting poisonous flow entries. Similarly, a malicious switch (e.g., Switch B) or hosts can intercept LLDP packets, filter them out, and relay or replay them to another switch (e.g., Switch C) to cause incorrect \textit{link-dst} (e.g. $A2\rightarrow1C$).  Our proposed defense to mitigate Marionette and other attacks, called \proj{}, is shown in Figure~\ref{fig:morphDisc}, which we discuss in Section~\ref{sec:insights}.

The \textbf{root cause} of topology poisoning attacks that rely on static identifiers is as follows: 

\begin{tcolorbox}[colback=cyan!5, colframe=cyan!80!black,left=3pt,right=3pt,top=3pt,bottom=3pt]
The static fields of LLDP packets (e.g., the LLDP \kw{ether-type}, which explicitly signals the discovery intent) are an easy target for attackers, who can spoof or exploit these fixed identifiers to manipulate packet forwarding via various attack surfaces. 
\end{tcolorbox}
\begin{table*}[!h]
\small

\caption{Effectiveness of State-of-the-Art Countermeasures against Topology Poisoning Attacks}
\label{tab:link_fab_contermeasure}
\vspace{-5pt}
%\vspace{-0.5\baselineskip}
\footnotesize
\setlength{\tabcolsep}{2.25pt}
\centering 
\begin{tabular}[htbp]{l|l|lllll|llllll|llll}
\toprule
\multicolumn{2}{l|}{\textbf{Topo. Poisoning Attacks}}
& \multicolumn{5}{l|}{\textbf{Open-Source Controllers}}
 & \multicolumn{6}{l|}{\textbf{Topology Poisoning Detections}}  
 & \multicolumn{4}{l}{\textbf{Secure SDN Discovery Protocols}}\\
 \midrule
\textbf{Cause} & \textbf{Attack Surface} 
& 
\textbf{Ryu} & \textbf{Pox} & \textbf{Flood} & \textbf{ODL} & \textbf{ONOS}
&
\textbf{Authen.} & \textbf{Port-} & \textbf{Topo} & \textbf{Latency}  & \textbf{Sphinx} & \textbf{SPV}

& 
\textbf{sOFTDP} &\textbf{SLDP} &\textbf{TILAK} &\textbf{\small\textsc{Chameleon}} 
\\

&  & 

& &  \textbf{light}&  &  &
\textbf{-Based}& \textbf{Based} & \textbf{Guard} & \textbf{-Based}  &  & 
&  & & & \textbf{\small\textsc{Disc}}
\\
\midrule

Spoof &  Mal-Host~\cite{hong2015poisoning} & 
\ding{55} & \ding{55} & \ding{51} & \ding{51} & \ding{51} & 
\ding{51} & \ding{51} & \ding{51} & \ding{51} & \ding{51} & \ding{51} & 
\ding{51} & \ding{51} & \ding{51} & \ding{51}
 \\
%\midrule
& Mal-Switch~\cite{baidya2020link} &  
\ding{55} & \ding{55} & \ding{55}  & \ding{55} & \ding{55} &
\ding{55} & \ding{51} & \ding{55} & \ding{51}   & \ding{51} & \ding{51} &
\ding{55}  & \ding{55} & \ding{55} & \ding{51} 
 \\
\midrule

Replay &  Mal-Host~\cite{nehra2019sldp} & 
\ding{55} & \ding{55} & \ding{55} & \ding{51} & \ding{51} & 
dyn. \ding{51} & \ding{55} & \ding{51} & \ding{51} & \ding{51} & \ding{51} &
\ding{51} & \ding{51} & \ding{51} & \ding{51} 
 \\
 &  & 
 &  &  & & & 
stat. \ding{55} & & &  &  &  &
 &  &  & 
 \\
%& Mal-Switch &
%\ding{55} & \ding{55} & \ding{55} & \ding{51} & \ding{51} & 
%dy. \ding{51} stat. \ding{55} & \ding{51} & \ding{51} & \ding{51} & \ding{51} & \ding{51} &
%\ding{55} & \ding{51} & \ding{51} & \ding{51} 
% \\
\midrule

Relay 
%& Mal-Host &
%\ding{55} & \ding{55} & \ding{55} & \ding{55} & \ding{55} & 
%\ding{55} & \ding{51} & \ding{51} & \ding{51} & \ding{51} & \ding{51} &
%\ding{51} & \ding{51} & \ding{51} & \ding{51} 
% \\
& Mal-Switch~\cite{alimohammadifar2018stealthy} &
\ding{55} & \ding{55} & \ding{55} & \ding{55} & \ding{55} & 
\ding{55} & \ding{55} & \ding{55} & \ding{55} & \ding{51} & \ding{51} &
\ding{55} & \ding{55} & \ding{55} & \ding{51} 
 \\
\midrule

Flow  & Mal-App~\cite{chen2024manipulating} &
\ding{55} & \ding{55} & \ding{55} & \ding{55}  & \ding{55}  & 
\ding{55} & \ding{55} & \ding{55} & \ding{55}  & \ding{55}  & \ding{55} &
\ding{55} & \ding{55} & \ding{55} & \ding{51} 
 \\
Entry & Mal-Ctrl Peer~\cite{chen2024manipulating} &
\ding{55}  & \ding{55} & \ding{55} & \ding{55} & \ding{55}  &
\ding{55} & \ding{55} & \ding{55} & \ding{55}  & \ding{55}  & \ding{55} &
\ding{55} & \ding{55} & \ding{55}  & \ding{51} 
 \\
\toprule
%\multicolumn{2}{l|}{\textbf{Security Overhead}} 
%& Calculate & Monitor   & Monitor  & Cal. \& Authen.  & Mon. OF & Mon.  &  NA  & NA  & Calculate & C.\&A.  & C.\&A.  & Modify  & Rand.  & Rand.  & Rand. Disc. \\

%&  & \& Authen.  & Port &Latency & Sig. \& Mon.  &\& Data & Pkt-In  & & & \& Authen.  & Dy. & Dy. & Switch  & \kw{ether} & \kw{ether} & Pkt \& Table \\

%& & Signature & Status & &  Port Status & Plane & \&Rules & & & Static Sig. & Sig. &Sig. & \& Ctrl. & \kw{-src} & \kw{-dst} &  Lookup \\
%\toprule
\end{tabular}
\vspace{-10pt}
\end{table*}

We define four types of topology poisoning attacks below:

\inlinedsectionbf{Spoof-Based Topology Poisoning} A malicious host or switch can fabricate LLDP or \kw{packet-in} that encapsulate LLDP, causing controller to infer incorrect \textit{link-src}~\cite{hong2015poisoning} or \textit{link-dst}~\cite{baidya2020link}.

\inlinedsectionbf{Replay-Based Topology Poisoning}~%\footnote{Some literature refers to this as host-replay topology poisoning attacks.}}
A malicious host stores LLDP packets and replays them to another switch, causing the legitimate controller to infer a %wrong \textit{link-dst}
nonexistent link~\footnote{The link discovery discussed refers to switch-to-switch links.}~\cite{nehra2019sldp}.

\inlinedsectionbf{Relay-Based Topology Poisoning} A malicious switch intercepts LLDP packets and relays them to another switch, causing the legitimate controller to infer a wrong \textit{link-dst}~\cite{alimohammadifar2018stealthy}.

\inlinedsectionbf{Flow Entry-Induced Topology Poisoning} A malicious application or controller peer injects poisonous flow entries to override default entries to manipulate LLDP packet forwarding, causing the legitimate controller to infer a wrong \textit{link-dst}~\cite{chen2024manipulating}.

\subsection{Existing Defenses and Challenges}

Existing OFDP implementations differ in packet encoding and forwarding rules, but all retain recognizable LLDP-based discovery identifiers (Appendix~\ref{appendix:ofdp_impl}). 
In addition to the efforts by open-source controllers discussed in Appendix~\ref{appendix:ofdp_impl} to secure SDN topology discovery, researchers have proposed approaches aimed at topology poisoning detection and the development of secure SDN discovery protocols. Table~\ref{tab:link_fab_contermeasure} summarizes the effectiveness of existing countermeasures against various types of topology poisoning attacks.

\emph{Detection approaches}  primarily rely on known attack signatures and are therefore effective against previously identified topology poisoning attacks. 
For example, authentication-based schemes verify LLDP integrity to prevent spoofing and replay attacks, but offer limited protection against other vectors. Port-based verification~\cite{dhawan2015sphinx, baidya2020link} assumes each active port should map to only one peer, catching spoof-based links but not others. Latency-based approaches~\cite{smyth2017detecting, huang2020towards} infer fabrication by measuring anomalously high link latencies, particularly effective when hosts are involved. TopoGuard and TopoGuard+~\cite{hong2015poisoning, skowyra2018effective} track switch identities to defend against host-involved spoofing but do not address control plane misuse. 
A more advanced class of detection methods includes \emph{verification-based approaches}, which aim to detect fabricated links by validating network state through auxiliary observations. The most capable solutions, such as SPHINX~\cite{dhawan2015sphinx} and SPV~\cite{alimohammadifar2018stealthy}, can detect most of the topology poisoning attacks.
However, they suffer from scalability issues due to real-time monitoring on the data plane and fail to detect recent flow entry–induced topology poisoning attacks, as they overlook the control-plane attack surface~\cite{chen2024manipulating}.

\emph{Prevention methods} aim to harden the discovery protocol against tampering. Many controllers embed static or dynamic LLDP secrets, but reuse of Ethernet headers reveals packet intent, making them vulnerable to switch-spoofing, relaying, and flow entry–induced attacks. \textsc{SLDP}\cite{nehra2019sldp} introduces partial header randomization with token-based source MACs and flow entries tied to tokens. However, other static header fields remain recognizable and can still be exploited by malicious switches, controller applications, or peers to manipulate packet forwarding. \textsc{TILAK}\cite{nehra2019tilak} offers similar security scope by randomizing only \texttt{ether-dst}, and thus faces the same limitations as \textsc{SLDP}. \textsc{sOFTDP}~\cite{azzouni2018softdp} uses encryption and Bidirectional Forwarding Detection (BFD) but requires switch modifications and incurs high deployment cost.

\inlinedsectionbf{Challenges} Preventing topology poisoning in SDN remains a significant challenge due to the architecture’s open and programmable nature. While centralized control offers powerful capabilities, such as global network visibility and real-time reconfigurability, it also introduces additional attack surfaces. Moreover, the inherent programmability of SDN expands the potential attack vectors. The growing diversity of topology poisoning attacks highlights a critical need: to maintain a secure and trustworthy topology view, the controller must operate under the assumption that no other network component, including hosts, switches, applications, or controller peers, can be fully trusted. Although several schemes mitigate specific attacks once the attack vector is identified, achieving comprehensive prevention remains difficult. Designing a mechanism that not only prevents topology poisoning but also detects ongoing manipulation—while remaining lightweight and SDN-compatible—poses an even greater challenge.

A key distinction underlies this challenge: prevention versus detection. \emph{Prevention-oriented} protocols harden the discovery process so attacks cannot succeed in the first place, while \emph{detection-oriented} approaches monitor for anomalies after discovery has occurred. 
As Table~\ref{tab:link_fab_contermeasure} shows, existing secure SDN discovery protocols (e.g., SLDP, TILAK) focus on prevention but cannot detect ongoing manipulation, while detection approaches
(e.g., SPHINX, SPV) react to anomalies but cannot prevent attacks from influencing the topology view in the first place. \emph{\proj{} is the first protocol to systematically combine defenses that prevent attacks during link discovery and detect
manipulations that evade those preventative measures}: a \emph{moving-target} discovery process built on dynamic randomization of all discovery-packet fields blocks identifier-based attacks during discovery, while verification triggered on each link change event catches manipulations that bypass randomization.

\section{Scope and Threat Model}~\label{sec:threat_model}
\inlinedsection{Scope} \proj{} is a lightweight, SDN-compatible protocol for discovering the topology of pure OpenFlow networks while preventing and detecting topology-poisoning attacks\footnote{Hybrid deployments with both traditional and OpenFlow switches are outside the scope of this work.}. We cover topology-poisoning attacks that exploit static LLDP identifiers --- including Spoof-, Replay-, Relay-, and Flow Entry-Induced attacks --- as well as advanced attacks that identify discovery signatures through traffic analysis. Here, \emph{discovery signatures} refers to recognizable packet types specifically used for link discovery; attacks that manipulate prevalent data-plane packet types for non-discovery purposes are out of scope.
We assume that no flow entry matches solely on the incoming port (\texttt{in\_port}): matching only on \texttt{in\_port} provides coarse control, is rarely used in production, and contradicts SDN best practices~\cite{ONFSpec1.5.1,kreutz2014software,ONFTR521}. We also assume SDN-prevalent reactive forwarding~\cite{mckeown2008openflow,kreutz2014software}, meaning a table-miss entry is installed to send unknown packets to the controller.

\inlinedsection{Threat Model} 
We assume only the leader controller responsible for topology discovery is trusted. All other components are considered potentially compromised. These attacks may originate from malicious hosts~\cite{hong2015poisoning,nehra2019sldp}, applications or controller peers~\cite{chen2024manipulating}, or switches~\cite{baidya2020link,alimohammadifar2018stealthy}. Specifically, we assume: 
\begin{itemize}[leftmargin=15pt, itemsep=0pt, parsep=0pt, topsep=2pt, partopsep=2pt] 
\item A malicious host may analyze, fabricate, or replay discovery packets (e.g., LLDP) toward its connected switch. 
\item A malicious application or controller peer may read topology, nodes, and flow entries, and write flow entries. 
\item A malicious switch may inspect packets to recognize discovery packets (including randomized variants), forward such packets between switches with or without flow entries, and tamper with their corresponding \kw{packet-in} messages.
\end{itemize}

Among these, the malicious switch is the most powerful: it has direct data-plane access, can manipulate any packet crossing it --- not only those tied to specific flow entries or known discovery formats --- and operates at hardware forwarding speeds that defeat latency-based detection. \proj{}'s defense layers scale with adversary capability --- DecoyType and MorphType address host, control-plane, and switch attacks on static-identifier discovery packets, while CamoType is the additional layer to thwart an advanced malicious switch that redirects discovery packets outside the flow-entry mechanism or applies traffic analysis to randomized discovery packets.

\inlinedsection{Adversary Knowledge}
The adversary knows the full \proj{} design, including that
discovery-packet fields are randomized using a cryptographically secure PRNG
(§\ref{sec:pkt_gen}), but not the generator's private state. A malicious host
observes only its own link, whereas a malicious switch may observe and tamper
with all traffic traversing it, including table-miss packets sent to the
controller. A malicious application or controller peer may additionally observe
\kw{packet-in} events. We reasonably assume that no single compromised host/switch observes the majority of hosts, since hosts are
distributed across access links and switches see only traffic traversing them.
The trusted leader controller shares neither its PRNG state, the identity of
future discovery probes, nor its private packet-to-source mapping with any other
network component. Therefore, even after observing many
\kw{packet-in}/\kw{packet-out} samples, an adversary lacking the private PRNG
state cannot predict future randomized probe fields. A deployment that
substitutes a weaker generator forfeits this guarantee.

\inlinedsection{Stealthy Adversary Assumption} Consistent with our focus on covert topology poisoning that leaves other services intact (§Scope), we impose two behavioral constraints. First, adversaries do not manipulate normal data-plane traffic (e.g., ARP or IP), as this disrupts operations and exposes their presence. Second, adversaries do not alter default flow entries such as the table-miss — not a limit on what a compromised switch can do, but on what a stealthy attacker will do: table-miss drives reactive forwarding, so overriding it stalls flow setup for every host behind the switch and produces controller-visible outages, and can further be locked by configuration~\cite{khurshid2012veriflow,kazemian2013real}. This does not exempt a stronger attacker who installs higher-priority rules to catch discovery packets; \proj{} defends against that case via field randomization and CamoType (§IV). The result is an inescapable dilemma regardless of privilege: disrupting normal forwarding risks exposure, while leaving it intact renders topology manipulation detectable.

\section{\proj{}: Less is More }\label{sec:insights}

To defend against topology poisoning attacks that exploit static identifiers or discovery signatures~\cite{hong2015poisoning,baidya2020link,nehra2019sldp,alimohammadifar2018stealthy,chen2024manipulating}, \proj{} addresses the \emph{root cause} (\S\ref{sec:motivation}) of these attacks in modern SDN discovery protocols. Our approach is guided by three key insights.

% To secure the topology view of the controller against existing topology poisoning attacks that exploit static identifiers or discovery signatures~\cite{hong2015poisoning,baidya2020link,nehra2019sldp,alimohammadifar2018stealthy,chen2024manipulating,chen2024manipulating}, we introduce \proj{}, a flexible and stealthy discovery methodology that addresses the \textbf{root causes} of existing topology poisoning attacks in recent SDN discovery protocols.  Our approach is informed by three key insights.

% \textcolor{red}{Insight descriptions could be significantly shortened.}

\begin{tcolorbox}[colback=cyan!5, colframe=cyan!80!black,width=0.98\linewidth,left=3pt,right=3pt,top=3pt,bottom=3pt]
\textbf{Insight 1:} SDN does not require informative discovery packets to reveal \kw{link-src} for link discovery. 
\end{tcolorbox}

In topology discovery protocols, a link is inferred when a packet sent from one switch port is received by another. In traditional distributed networks---where switches operate independently and maintain local topology state---each switch must self-identify in its LLDP packet via a \kw{Chassis ID} and \kw{Port ID}. This allows neighboring switches to record the connection and identify the link (Appendix~\ref{appendix:lldp_detail}).

By contrast, SDN’s centralized architecture gives the controller complete knowledge of all switches once controller–switch connections are established. Embedding explicit \textit{link-src} information in discovery packets is therefore unnecessary and instead leaks sensitive details that attackers can exploit. We observe that neither the payload nor the header must reveal \textit{link-src}. Because the controller acts as both sender and receiver of discovery packets, it can maintain a local mapping between each transmitted packet and its corresponding \textit{link-src}. When a packet returns, the controller consults this mapping to recover the \textit{link-src}, enabling accurate topology reconstruction without embedding sensitive information in the packet. We assign each link-src a distinct randomized packet, ensuring that returning packets map unambiguously back to their source. %\trent{Is this possible if multiple discovery packets are sent in parallel?}\mm{yes, it is possible but won't impacet the correct link-src recovery because we send differently randomized packets for each port especially when sending them in parallel. do I need to emphasize that?}

Recent work such as SLDP~\cite{nehra2019sldp} notes that many LLDP fields are unnecessary in SDN and trims the format to an Ethernet header plus \kw{Chassis ID} and \kw{Port ID}. However, these fields remain sensitive: adversaries can still exploit them to spoof discovery packets and fabricate false links.

\begin{tcolorbox}[colback=cyan!5, colframe=cyan!80!black,width=0.98\linewidth,left=3pt,right=3pt,top=3pt,bottom=3pt]
\textbf{Insight 2:} Dynamically randomizing all discovery-packet fields 
%under specific conditions 
can preserve link discovery functionality while preventing information leakage, blocking static-identifier attacks, and enabling camouflage.
\end{tcolorbox}

As shown in Figure~\ref{fig:motivating_example}, static fields such as \texttt{ether-dst}, \texttt{ether-src}, and \texttt{ether-type} expose recognizable signatures that attackers can exploit to manipulate topology discovery. In traditional networks, discovery protocols rely on fixed packet formats and protocol-specific handling. In contrast, SDN provides an architectural advantage that \proj{} leverages to eliminate such vulnerabilities.

We observe that reactive forwarding---a key capability of SDN---creates an opportunity to secure SDN discovery. In SDN, each switch installs a table-miss flow entry that forwards all unmatched packets, regardless of their headers, to the controller. This mechanism provides an opportunity for discovery packets to reach the controller after exactly one link traversal without relying on flow entries that match a static \texttt{ether-type} or other fixed header fields. Essentially, any discovery packet that triggers the table-miss rule after one-link hop can serve as a discovery packet. Consequently, \proj{} gains the flexibility to select discovery packets from a wide range of packet types and fields while maintaining correct link inference, thereby eliminating the fundamental dependency on static identifiers that traditional LLDP-based protocols require and enabling security mechanisms to be naturally embedded in the discovery process.

\begin{tcolorbox}[colback=cyan!5, width=0.98\linewidth, colframe=cyan!80!black,left=3pt,right=3pt,top=3pt,bottom=3pt]
\textbf{Insight 3:} 
Existing defenses assume discovery packets must self-identify, inherited from distributed networks. Yet, centralized SDN architectures can remove this requirement, enabling fully controller-driven discovery.
%Existing defenses inherit distributed-network assumptions that discovery packets must be self-identifying, while the centralized SDN architecture eliminates this requirement and enables a fully centralized controller-driven discovery process.
\end{tcolorbox}
Prior approaches failed to fully exploit this architectural shift because they missed its key implication. SLDP introduces token-based MAC randomization for authentication but still relies on static identifiers, such as a fixed \kw{ether-type}, to mark discovery packets. This dependency leaks a stable discovery signature, leaving SLDP vulnerable to relay-based poisoning by malicious switches and flow entry–induced poisoning by malicious applications or controller peers.

% Prior approaches did not fully embrace this architectural shift because they overlooked this key insight.
% SLDP introduces token-based source MAC randomization for authentication but still relies on static identifiers—such as the \kw{ether-type} to identify discovery packets.
% This dependence exposes fixed discovery signatures, leaving SLDP vulnerable to relay-based topology poisoning by malicious switches, or flow entry–induced topology poisoning by a malicious application/controller peer.

SPV, in contrast, depends on the SDN controller to discover links and performs only post-hoc verification using IP probes. It reacts to link-change events reported by the controller—events derived from identifier-based LLDP discovery, and therefore can detect but cannot prevent topology-poisoning or proactively discover links. %\trent{SPV cannot discover links - isn't that fundamental to link discovery?  I assume I am missing something.}\mm{SPV is a detection system augmented to SDN controller, it retrieve link discovery \& link change information from the controller to verify link change with IP packets. It assumes the control plan is benign. Do I need to clarify something?} 
SPV also assumes a trusted control plane, making it unable to detect flow entry–induced poisoning launched by a malicious application or controller peer.

% SPV, in contrast, depends on conventional SDN discovery protocols and only conducts after-event verification using IP-based probes. It reacts to link-change events reported by the controller, which itself depends on a traditional identifier-based LLDP discovery process, and therefore detects but cannot prevent topology-poisoning attacks or proactively discover links. Moreover, SPV assumes a trustworthy control plane. A malicious application or controller peer can launch flow-entry–induced topology poisoning, which SPV fails to detect.

%\textcolor{red}{TJ: This is probably better than the last paragraph of Insight 2.}
% Neither SLDP nor SPV recognizes or builds on the key insights.
% Both inherit the traditional notion that discovery packets must disclose their source identity and retain a static format recognizable on the wire.
% \proj{}, by contrast, abandons this legacy assumption altogether. It redefines SDN discovery as an information-free, dynamically evolving process that removes every static identifier exploitable by an adversary and unifies prevention and detection in a single lightweight design.
% This architectural shift makes \proj{} a fundamentally new discovery paradigm—one that transforms randomization from a local authentication mechanism into a global defense strategy capable of neutralizing existing topology-poisoning vectors relying on static identifiers.

Neither SLDP nor SPV recognizes or leverages the central insight: both preserve the legacy assumption that discovery packets must reveal their source identity and follow a fixed, wire-recognizable format. \proj{} discards this assumption entirely. 
It redefines SDN discovery as an information-free, dynamically evolving process that eliminates all static identifiers exploitable by adversaries, unifying prevention and detection within a single lightweight design that is
capable of neutralizing all identifier-based topology-poisoning vectors.

\proj{} prevents topology-poisoning attacks~\cite{chen2024manipulating,hong2015poisoning,baidya2020link,nehra2019sldp,alimohammadifar2018stealthy} by removing static identifiers that attackers can use to manipulate the discovery packet forwarding. As shown in Figure~\ref{fig:morphDisc} (see page 3), with the \texttt{ether-src} randomized in a {MorphType} packet, $e^{B}2$ no longer matches the packet after Step~\ding{193}. Consequently, at Step~\ding{194}, Switch~B behaves normally, returning the discovery packet to the controller due to the table-miss. At Step~\ding{195}, the controller consults its mapping table to recover the \texttt{link-src} ($A2$), allowing the controller to identify the legitimate link $A2\rightarrow1B$. 
{CamoType} packets use ARP to follow the same path for verification.

\section{System Design}
\label{sec:design}

Building on these insights, we observe that any packet---regardless of header or payload---can support link discovery in SDN as long as it triggers the table-miss rule after one hop. \proj{} secures discovery through a three-layer defense, each layer strengthening protection against progressively more capable adversaries. \emph{DecoyType} retains traditional LLDP-based discovery as an intentional honeypot (\textbf{Insight 3}): attackers who exploit it reveal themselves through inconsistency, while those who ignore it are caught by deeper layers. Because DecoyType reuses existing LLDP infrastructure, including it costs nothing even when aware attackers ignore it. \emph{MorphType} eliminates informative payloads through full randomization (\textbf{Insights 1 and 2}): the controller's mapping table recovers \textit{link-src} without embedding it in the packet, and removing static identifiers prevents exploitation. \emph{CamoType} leverages camouflage as normal traffic for selective verification (\textbf{Insight 2}): under the stealthy adversary model, tampering with such normal data-plane packets risks immediate exposure. Each layer is detailed below.

\inlinedsectionbff{DecoyType} packets preserve traditional identifiers and serve as bait for attackers exploiting static, LLDP-style discovery. Their predictable formats and explicit \kw{link-src} fields attract legacy attacks, which the controller later detects by cross-checking inconsistent results.

\inlinedsectionbff{MorphType} packets randomize all headers and payloads to eliminate identifier and \kw{link-src} leakage. Running in parallel with DecoyType, MorphType provides clean, identifier-free discovery. Any discrepancy between the two immediately signals manipulation of discovery traffic.

\inlinedsectionbff{CamoType} packets act as a final, camouflaged verification layer. Whenever DecoyType or the verified MorphType result reports a link change --- i.e., a discovery result differing from the last-confirmed topology --- the controller issues CamoType ARP probes constructed with fresh, non-conflicting MAC/IP identities that mimic legitimate new-host arrivals and blend with normal ARP traffic on the wire (Section~\ref{sec:design}). Under the stealthy adversary model (Section~\ref{sec:threat_model}), selectively manipulating such traffic risks disrupting legitimate ARP processing and exposing the attacker.

\begin{table}[t]
\centering
\begin{threeparttable}
\caption{Deployment Prerequisites for \proj{}}
\label{tab:prerequisites}
\small
\begin{tabular}{p{4.8cm}p{0.7cm}p{0.7cm}p{0.7cm}}
\toprule
\textbf{Requirement} & \textbf{Decoy} & \textbf{Morph} & \textbf{Camo}\\
\midrule
(a) Default flow entries are protected & \ding{51} & \ding{51} &\ding{51}  \\
\midrule
(b) A default flow entry forwarding LLDP/BDDP packet to controller & \LEFTcircle & --- & ---\\
\midrule
(c) A default table-miss entry forwarding unmatched packet to controller & \LEFTcircle & \ding{51} & ---\\
\midrule
(d) A default flow entry forwarding ARP packet to controller & --- & --- & \ding{51}\\
\midrule
(e) No flow entry matching solely on \kw{in\_port} intercepting table-miss & \LEFTcircle & \ding{51} & ---\\
%\midrule
%(f) Normal data-plane traffic is not manipulated by stealthy adversaries & --- & \ding{51} & \ding{51}\\
\bottomrule
\end{tabular}
\begin{tablenotes}
\item[\ding{51}] Required; --- Not required;
\LEFTcircle DecoyType requires either (b), or both (c) and (e), depending on the controller, see Table~\ref{tab:flow_entries4discovery}.
\end{tablenotes}
\end{threeparttable}
\vspace{-8pt}
\end{table}

\proj{} requires the following conditions for each of DecoyType, MorphType, and CamoType to operate correctly, as shown in Table~\ref{tab:prerequisites}. It uses DecoyType and MorphType for periodic link discovery, and reserves CamoType for link change verification. This separation is motivated below:
\begin{itemize}[leftmargin=15pt, itemsep=0pt, parsep=0pt, topsep=2pt, partopsep=2pt]
\item \textbf{MorphType reinforces DecoyType:} MorphType performs periodic discovery without static identifiers, preventing topology poisoning attacks that rely on fixed signatures. %meanwhile, \update{any remaining inconsistency after MorphType verification constitutes a link-change event, which CamoType resolves.}
\item \textbf{CamoType defends against MorphType evasion:} if an advanced attacker intercepts all unknown packets to classify and attack MorphType, CamoType — disguised as normal data plane traffic — creates an inescapable dilemma: manipulating it disrupts normal communication and immediately exposes the attacker, while avoiding it leaves the topology manipulation detectable. %\proj{} further frustrates classification by \update{generating standards-compliant ARP probes that mimic new-host arrivals and introducing a randomized verification delay, so selectively targeting CamoType requires distinguishing it from legitimate ARP activity.}
%reusing genuine host-generated ARP packets and introducing a randomized verification delay, so CamoType cannot be selectively targeted without also disrupting legitimate ARP traffic. 
\item \textbf{CamoType is reserved for verification only:} overusing CamoType (e.g., ARP/IP) for periodic discovery would leak too many samples, helping attackers distinguish discovery probes from legitimate traffic. By limiting CamoType to link change events — which are infrequent — \proj{} preserves camouflage and denies attackers the sample volume needed for classification.
\end{itemize}

\subsection{Built-in Prevention and Detection}  
\label{sec:prev-detect}

\begin{algorithm}[t]
\caption{Prevention and Detection of \proj{}}\label{alg:verfi_detect}
\begin{algorithmic}[1]
\footnotesize
    \REQUIRE $DS$: Topo datastore;
    $P$: Ports; $T$: Mapping table; $q$: Num of rounds.
    \WHILE{true}
        \FOR{$p_i \in P$}
        \STATE $D.pkt_i$ $\gets$ \textbf{gen\_Dpkt}($p_i$); \quad\quad
        $M.pkt_i$ $\gets$ \textbf{gen\_Mpkt}($p_i$);
        \STATE \textbf{upd\_Map}($T$, $D.pkt_i$, $p_i$); \quad\quad \textbf{upd\_Map}($T$, $M.pkt_i$, $p_i$);
        \STATE $D.pktIn_i$ $\gets$ \textbf{probe}($D.pkt_i$);
         $M.pktIn_i$ $\gets$ \textbf{probe}($M.pkt_i$);
        \STATE $D.link_i$ $\gets$ \textbf{inf}($D.pktIn_i$,$T$);\quad
         $M.link_i$ $\gets$\textbf{inf}($M.pktIn_i$,$T$);
        \STATE $M.chg_i$ $\gets$ \textbf{diff}($DS, M.link_i)$;
        \IF{$|M.chg_i| > 0$}
            \STATE $M.link_i$ $\gets$ $\textbf{multi\_round\_verify}(M.chg_i, p_i, q)$;
            %\STATE \textbf{report}(randGuess or networkFault);
        \ENDIF
        \STATE $Chg.link_i$ $\gets$
        \textbf{diff}($DS, D.link_i, M.link_i$);
        \IF{$|Chg.link_i| > 0$}
            \STATE $R.sw$ $\gets$ \textbf{get\_sw}($Chg.link_i$);
            \STATE $C.pkt_i$ $\gets$ \textbf{gen\_Cpkt}($p_i$);
            \STATE \textbf{upd\_Map}($T$, $C.pkt_i$, $p_i$);
            \STATE $C.pktIn_i$ $\gets$ \textbf{probe}($C.pkt_i$);
            \STATE $C.link_i$ $\gets$ \textbf{inf}($C.pktIn_i$,$T$);
            \IF{$D.link_i == M.link_i == C.link_i$}
                \STATE \textbf{update}($DS, Chg.link_i$);
                \STATE \textbf{report}(AcceptLinkChg);
            \ELSIF{$M.link_i == C.link_i$}
                \STATE \textbf{update}($DS, C.link_i$);
                \IF{\textbf{detect\_pois\_Entry}($R.sw$)}
                    \STATE \textbf{report}(Marionette);
                \ELSIF{\textbf{detect\_host}($R.sw$)}
                    \STATE \textbf{report}(MalHost);
                \ELSE
                    \STATE \textbf{report}(MalSw);
                \ENDIF
            \ELSE
                \STATE \textbf{report}(AdvAttack);
            \ENDIF
        \ENDIF
        \ENDFOR
    \ENDWHILE
\end{algorithmic}
\end{algorithm}

\begin{figure}[t]
\centering{\includegraphics[width=0.2\textwidth]{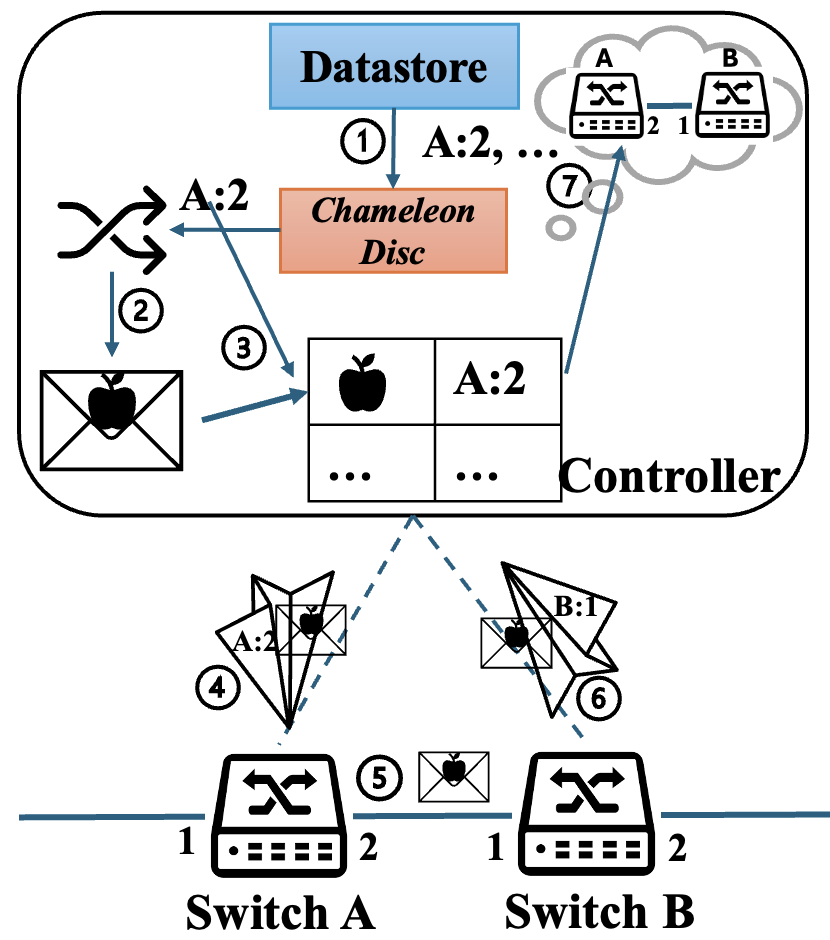}} 
    \caption{\proj{} System Design}    
    \label{fig:workflow}
    \vspace{-15pt}
\end{figure}

\proj{} embeds prevention and detection directly into its discovery workflow, allowing the controller to continuously discover network topology and simultaneously prevent and detect topology poisoning attacks as part of a single, integrated process. To ensure robust protection along with detection against forged links, we layer three discovery probes with progressively higher security levels and trigger CamoType verification on every link change event --- i.e., whenever any probe reports a discovery result differing from the last-confirmed topology --- accepting the link update only when CamoType corroborates the result. Algorithm~\ref{alg:verfi_detect} formalizes this workflow, where lines~3--6 correspond to the DecoyType and MorphType discovery phases and lines~10--13 correspond to the CamoType verification phase. The unified workflow encompassing all three phases is illustrated in Figure~\ref{fig:workflow}.

Figure~\ref{fig:workflow} and Algorithm~\ref{alg:verfi_detect} together illustrate the full discovery process. At step \ding{192}, \proj{} collects the current topology and switch-port information for each discovery cycle. At step \ding{193}, for each port $p_i$ (line 2), the controller generates corresponding DecoyType and MorphType packets for $p_i$ (line 3). At step \ding{194}, the controller stores the mapping of DecoyType and MorphType to their \textit{link-src} (line~4). At steps \ding{195}\ding{196}\ding{197}, the DecoyType and MorphType packets are sent out as probes and returned to the controller through the table-miss flow entry (line 5). At step \ding{198}, the controller recovers the \textit{link-src} by querying the mapping table and derives the \textit{link-dst} from the \kw{packet-in} message (line 6). The steps of CamoType verification (lines~14--17) follow the same logic.
%\trent{I recall seeing a lower MB memory overhead earlier.}\mm{I am confused. Which section are you discussing? I don't see memory overhead in this section.}

Continuing in Algorithm~\ref{alg:verfi_detect}, when MorphType reports a link change (line~7), \proj{} triggers multi-round MorphType verification to filter out transient network faults and random-guessing-based flow entry-induced topology poisoning (lines~8--10). Whenever any probe (DecoyType or MorphType) reports a link change — i.e., a discovery result differing from the last-confirmed topology (lines~11--12), the controller triggers CamoType verification (lines~14--17) and resolves the outcome via three branches. First, if DecoyType, MorphType, and CamoType all agree on the new link (line 18), the link update is accepted as legitimate (lines~19--20). Second, if only MorphType and CamoType agree (line 21), DecoyType is the manipulated layer; the controller updates the topology with CamoType's report and inspects the responsible switch to classify the attack source as Marionette, MalHost, or MalSw (lines~22--28). Third, if MorphType's report does not match CamoType (line 30), the controller treats CamoType as the ground-truth reference --- justified by the Stealthy Adversary Assumption (\S\ref{sec:threat_model}), under which tampering with CamoType's normal data-plane probes would manipulate ARP traffic and risk immediate exposure --- and flags an advanced attacker (AdvAttack) capable of manipulating MorphType forwarding (line 31). We formalize this guarantee in Proposition~\ref{prop:camo}.

This verification is robust to benign faults because link update acceptance requires corroboration: an incorrect verdict needs a DecoyType, MorphType, and CamoType probes to fail in the same cycle, and the joint failure is far less likely than failure of a single probe~\cite{bolot1993end}. A transient loss is therefore resolved by re-probing before a verdict is finalized; a late-arriving probe behaves like a loss for that round, and \kw{packet-in} congestion manifests as added delay, so all such faults reduce
to the transient-loss case. This exposure is limited to begin with, as verification is triggered only on a reported change and inter-switch links are stable the vast majority of the time~\cite{gill2011understanding}. Finally, a sustained high joint loss rate is operationally equivalent to a link failure, so treating the link as unavailable is the correct outcome rather than a false positive. This robustness costs one extra CamoType verification per change, adding a brief confirmation delay. The exchange is favorable: the delay is paid only on rare link-change events, and a short wait to confirm a link is far cheaper than committing the controller to a false topology.

The discovery process for DecoyType packets follows traditional SDN mechanisms and introduces no additional overhead. The primary challenges lie in designing and composing  MorphType and CamoType packets.

\subsection{EtherType and Packet Construction}\label{sec:ether-type-selection}
Constructing discovery packets correctly is essential for both concealment and network interoperability. The \kw{ether-type} field is tightly regulated: specific values denote standardized packet formats whose payloads must follow corresponding specifications. \proj{} therefore selects \kw{ether-type} values that either (i) correspond to a packet type forwarded to the controller under SDN flow-entry configuration, or (ii) fall within IEEE-unreserved ranges~\cite{ieee_registration} suitable for triggering table-miss. Each discovery packet is padded to at least 64 bytes to meet the minimum Ethernet frame size.

Table~\ref{tab:ether-type} (Apppendix~\ref{appendix:ether-type}) summarizes the selected types. LLDP and BDDP~\cite{onos_disc,floodlight} serve DecoyType. Unregistered EtherTypes form the MorphType pool: only 276 of the $2^{16}$ values are registered, leaving ample space for randomized, nonconflicting packets. ARP (0x0806) is ideal for CamoType, as it is naturally forwarded to controllers via a single preinstalled flow rule.\footnote{match: \kw{ether-type}: 0x0806; action: \kw{to-controller}.} IPv4/IPv6 can also serve as CamoType but require more careful flow-table handling to ensure reliable table-miss behavior~\cite{alimohammadifar2018stealthy}.

\subsection{Non-Disruptive Verification}\label{sec:nondisruptive}

CamoType leverages normal data-plane packet types for verification, requiring careful construction to avoid interference with existing controller services. \proj{} therefore synthesizes standards-compliant ARP probes that mimic legitimate new-host arrivals using a fresh MAC address not present in the controller's host state and an unused IP address from the local subnet. This avoids reusing an existing host identity, which could otherwise cause ARP Handler or Host Tracker to associate that host with the wrong switch port.

A CamoType-specific filter prevents these controller-generated synthetic identities from being processed as legitimate hosts by ARP Handler and Host Tracker. Because the synthetic MAC/IP pair is fresh and does not overlap with an existing host, the filter can identify CamoType probes without suppressing legitimate ARP traffic. CamoType probes are additionally issued after a randomized delay to reduce timing-based classification. IP packets can also serve as CamoType probes, but require more careful flow-table handling to ensure reliable forwarding; we refer to SPV~\cite{alimohammadifar2018stealthy} for a detailed treatment.

\subsection{Preventing Table-Miss Overrides}
\label{sec:table-miss-override}
Because DecoyType (\textit{LLDP}) and CamoType (\textit{ARP}) packets are naturally sent to the controller by design, 
preventing table-miss overrides is only crucial for MorphType to maintain accurate link discovery. Our focus is on scenarios that reuse the default table-miss flow entries for discovery without dynamic data plane monitoring.
The \kw{EtherTypes} considered here include unassigned types (MorphType) and Layer-2 packet types of ARP (CamoType) and LLDP (DecoyType) that are inherently designed to be sent to the controller.
The Layer-3 \kw{EtherType} of IP (CamoType) follows the treatment described in SPV~\cite{alimohammadifar2018stealthy}, to which we refer for further details. Additionally, we select randomized MAC addresses that are not present in the SDN network for \kw{ether-src} and \kw{ether-dst}, further minimizing interference from existing flow entries and reducing the likelihood of unintended matches. This strategy eliminates the need for dynamic data plane monitoring and ensures that normal flow entries do not inadvertently match discovery packets for the reasons in Appendix~\ref{appendix:reason}.

If the SDN controller does not configure flow entries based on the Ethernet header, any MAC address in the \kw{ether-dst} or \kw{ether-src} fields may inadvertently trigger table-miss flow entries. 
When flow entries explicitly match \kw{ether-src}, \kw{ether-dst}, or \kw{ether-type}, \proj{} ensures that discovery packets avoid using those specific Ethernet headers to prevent unintended flow entries by avoiding all MAC addresses already present in the target SDN network, as normal flow entries do not match MAC addresses absent from the network. 
This approach is efficient, with its performance backed by probabilistic guarantees, as demonstrated next. 

\subsection{Efficient Ethernet Header Randomization}\label{sec:pkt_gen}

We apply Ethernet header randomization only to MorphType packets, as DecoyType and CamoType require registered packet types to fulfill their respective roles. 
Because certain MAC addresses must be excluded from the candidate pool, and the pool size is enormous, maintaining the entire candidate pool for shuffling randomized MAC addresses or EtherType values consumes significant memory. To address this issue, we opt to generate randomized numbers first and then verify whether they are allowed or blocked. Below, we provide evidence supporting the practicality of this method.

Generating collision-free discovery-packet fields is reliable and efficient, even at large network scale. In large SDN deployments, the total number of MAC addresses (across hosts and switch ports) is at most tens of millions~\cite{nunes2014survey,cisco_app_design,jain2013b4}, while the theoretical maximum is 
$2^{48}$ (about 281 trillion). Thus, the probability of randomly generating a MAC address that exists in the network is $3.55 \times 10^{-8}$. The probability of generating two MAC addresses that both collide with existing ones is merely $1.26025 \times 10^{-15}$. Similarly, randomly generating an \kw{ether-type} that falls among the 276 effectively registered out of $2^{16}$ possible values is 
0.00421, and regenerating a second collision is $1.77441 \times 10^{-5}$. Hence, randomly selecting \kw{ether-src}, \kw{ether-dst}, and \kw{ether-type} while checking against known addresses/types remains highly feasible. To efficiently handle repeated checks against known addresses and types, we use a hash set (preprocessed for fast lookup), which provides an average lookup time of $O(1)$, further reducing time complexity.

\begin{algorithm}[t]
\caption{MorphType Discovery Packet Generation}
\begin{algorithmic}[1]\label{alg:disc_pkt_gen}
\footnotesize
    \REQUIRE $x$: Switch ID,\\
    $macList$: List of MAC addresses existing in the target network, \\
    $typeList$: List of EtherTypes that should be avoided. 
    \ENSURE $\{morphPkt^x_i\}$: MorphType packets for Port $i$ on Switch $x$. \\
    %\STATE{$match_{ether}^x$: Match set of flow entries matching Ethernet header on Switch $x$}
    \FOR{Port $i$ on Switch $x$}
        \STATE $sMAC=dMAC=0$; 
        \STATE $morphType=0$;
        \STATE $macList = add(sMAC, macList)$;
        \STATE $typeList = add(morphType, typeList)$;
        \WHILE{$sMAC\in macList$}
            %\STATE{$sMAC =$ Randomly generated 48-bit MAC address}
            \STATE{$sMAC = rand\_gen\_mac()$; }
        \ENDWHILE
        \WHILE{$dMAC\in macList$}
            %\STATE{$sMAC =$ Randomly generated 48-bit MAC address}
            \STATE{$dMAC = rand\_gen\_mac()$; }
        \ENDWHILE
        \WHILE{$morphType\in typeList$}
            %\STATE{$sMAC =$ Randomly generated 48-bit MAC address}
            \STATE{$morphType = rand\_gen\_type()$; }
        \ENDWHILE
        \STATE{$morphHdr = build\_hdr(sMAC$, $dMAC$, $morphType$)};
        \STATE $morphPkt^x_i = build\_disc\_pkt(morphHdr)$;
    \ENDFOR
\end{algorithmic}

\end{algorithm}

Building on the discussion of packet construction and table-miss enforcement,
we propose an algorithm (Algorithm~\ref{alg:disc_pkt_gen}) to generate \proj{} discovery packets. For each switch, packets are randomly generated for every port (lines 1--17). Each port starts with source MAC, destination MAC, and \kw{ether-type} set to 0 (lines 2--3), which are then added to the MAC address list and type list (lines 4--5). While any value of \kw{ether-src}, \kw{ether-dst}, and \kw{ether-type} remains in its respective list, it is regenerated randomly until it no longer appears (lines 6--14). The random regeneration 
draws from the operating system's cryptographically secure random source (/dev/urandom)~\cite{gutterman2006analysis}, built on a NIST-standardized DRBG design~\cite{2811}, so generated values remain unpredictable to an adversary even after observing many packet-in/packet-out samples. Finally, we construct the Ethernet header, build the discovery packet, and add it to the packet list (lines 15--17). The MAC list is sourced from the controller's host tracker, which maintains active host and switch-port MACs as part of standard SDN operation and evicts entries when hosts expire. Because the host tracker is internal to the trusted controller, attackers cannot manipulate this list to shrink or enlarge the exclusion pool. 

%In the next section, we present a theoretical analysis quantifying the resilience of this design.

\section{Theoretical Resilience Analysis}
\label{sec:resilience}

Although MorphType's defense relies on randomized discovery packets and is therefore inherently probabilistic, this section establishes the theoretical foundation for evaluating its resilience against random guessing under realistic conditions. 
We adopt a $q$-round MorphType verification model corresponding to Algorithm~~\ref{alg:verfi_detect}. Multi-round verification is triggered only when MorphType reports a change, thereby retaining lightweight single-round probing during normal operation while strengthening verification of suspicious or changed topology observations. CamoType provides a separate final verification layer, whose guarantee is formalized in Proposition~\ref{prop:camo}.

By incorporating practical parameters aligned with real-world discovery settings, we show that the probability of a successful topology poisoning attack is extremely low. Moreover, multi-round
verification further diminishes this likelihood, while the inclusion of CamoType final verification (e.g., ARP) elevates
the defense to an additional level --- since ARP is a ubiquitous and legitimate packet type that adversaries are highly unlikely to manipulate without exposing their presence; we formalize this guarantee in Section~\ref{sec:camoguarantee}. Together, these mechanisms provide strong assurance of \proj{}'s robustness even in adversarial environments.

Our objective is to bound MorphType resilience to random guessing — i.e., the expected number of rounds before an attacker can force MorphType to report a fabricated link via random flow-entry installation. End-to-end correctness is provided by CamoType verification on every link change event (Proposition~\ref{prop:camo}, §\ref{sec:camoguarantee}); the bound below quantifies how rarely MorphType fabrications even reach that gate. This bound covers attackers restricted to flow entry installation; compromised switches that bypass the flow-entry mechanism (e.g., advanced relay) fall outside MorphType's defense and are handled separately by CamoType (Proposition~\ref{prop:camo}). %\trent{Even if they have compromised a switch? If we are talking about different attackers at different times, we need to be explicit about that.} \mm{a compromised switch can do LLDP relay or advanced randomized packets relay beyond flow entry constraint. Here is the traditional LLDP relay. not the advanced one.} 
The attacker is limited to maintaining at most $p$ malicious flow entries matching either of \kw{ether-src}, \kw{ether-dst}, and \kw{ether-type} in the flow table without detection. The SDN controller re-discovers links every $l$ seconds; when MorphType reports a candidate link change, the controller requires $q$ consistent MorphType discoveries before passing the result to final verification.

Before formalizing the analysis, we address two attack strategies that appear viable but are fundamentally ineffective.

\inlinedsection{Wildcard and bitmask matching} Wildcard matching ignores a field entirely; bitmask
matching matches arbitrary bit patterns via a mask~\cite{openflow}. Both mechanisms are typically used on IP addresses, where the hierarchical subnet structure (CIDR prefixes) makes selective
matching meaningful --- a structure absent in EtherType and MAC address spaces. Consequently, any mask broad enough to catch MorphType's randomized EtherTypes also catches standard data-plane types (IPv4, IPv6, ARP). Catching a uniformly random MAC similarly requires a mask covering a substantial fraction of the address space, which also covers legitimate host and switch MACs. In both cases, intercepting MorphType requires disrupting normal traffic, violating the Stealthy Adversary Assumption (\S\ref{sec:threat_model}). The following analysis therefore restricts the attacker to specific-value matching.

\inlinedsection{Adaptive observation} Even if an attacker observes a MorphType packet in transit and extracts its randomized header fields, two barriers prevent exploitation. First, installing a matching flow entry requires completing an OpenFlow rule installation cycle within the millisecond window before the controller collects the current probe — practically infeasible. Second, \proj{} re-randomizes all header fields every discovery round, so knowledge of round $N$ provides zero advantage for round $N+1$.
The attacker must successfully intercept $q$ consecutive independent rounds, the expected time for which is quantified below.

\subsection{MorphType Resilience Markov Chain}\label{sec:markov}
Consider a set containing \(2^m\) consecutive numbers:
\[
S = \{0, 1, 2, \ldots, 2^m - 1\}.
\]
An attacker randomly selects \(p\) numbers from this set. In each round, the SDN controller randomly selects a number from a subset \(T \subseteq S\) of size \(n\). 
We define:
\begin{itemize}[leftmargin=15pt, itemsep=0pt, parsep=0pt, topsep=2pt, partopsep=2pt]
    \item A \textbf{\(q\)-hit streak} is $q$ consecutive rounds where the controller’s choice lies within the attacker’s $p$ chosen values.
    \item A \textbf{\(q\)-miss streak} is $q$ consecutive rounds where the controller’s choice lies outside attacker’s $p$ chosen values.
\end{itemize}
The objective is to compute the expected number of rounds required to achieve a \(q\)-hit streak or a \(q\)-miss streak.

This problem can be modeled using a Markov chain, as the probability of the next round making \(q\)-hit streak depends only on the current state of \((q-1)\)-hit streak and not on the sequence of prior states. Without loss of generality, we focus on analyzing the probability of a \(q\)-hit streak; the analysis for a \(q\)-miss streak follows analogously.

%\subsubsection{\textbf{Probabilities}}
\begin{itemize}[leftmargin=15pt, itemsep=0pt, parsep=0pt, topsep=2pt, partopsep=2pt]
    \item Probability of the SDN controller selecting a number in the attacker's set:
    \begin{equation}
        P_{\text{in}} = \frac{p}{n}
    \end{equation}
    \item Probability of selecting a number outside attacker's set:
    \begin{equation}
    P_{\text{out}} = 1 - P_{\text{in}} = \frac{n - p}{n}
    \end{equation}
\end{itemize}

%\subsubsection{\textbf{Setup}}
Define states \(0, 1, \ldots, q\):
\begin{itemize}[leftmargin=15pt, itemsep=0pt, parsep=0pt, topsep=2pt, partopsep=2pt]
    \item State $i$ (\(S_i^H\)): The process has \(i\)-hit streak.
    \item State $0$ (\(S_0^H\)): No such streak currently exists.
    \item State $q$ (\(S_q^H\)): A \(q\)-hit streak has occurred.
\end{itemize}

Transitions:
\begin{itemize}[leftmargin=15pt, itemsep=0pt, parsep=0pt, topsep=2pt, partopsep=2pt]
    \item From State \(i < q\), moving to state \(i+1\) occurs with probability \(P_{\text{in}}\), and falling back to state \(0\) (breaking the streak) occurs with probability \(P_{\text{out}}\).
\end{itemize}

%\subsubsection{\textbf{Expected Number of Rounds}}
\begin{itemize}[leftmargin=15pt, itemsep=0pt, parsep=0pt, topsep=2pt, partopsep=2pt]
    \item \(E_i\): Expected number of rounds to reach state \(q\) from state \(i\).
\end{itemize}

Using Markov chain theory and the law of total expectation, the recurrence relations for expected hitting times are:
\begin{equation}\label{law_total_exp}
E_i = 1 + P_{\text{in}} E_{i+1} + P_{\text{out}} E_0, \quad \text{for } i < q
\end{equation}
with the boundary condition:
\begin{equation}
E_q = 0
\end{equation}

The boundary condition \(E_q = 0\) arises because once the process reaches state $q$, the $q$-hit streak is achieved, and no further rounds are needed to reach the desired goal.
Solving this system of equations yields \(E_0\), the expected number of rounds to achieve \(q\)-hit streak— an informative measure of \proj{}'s resilience.

We evaluate the MorphType resilience to random guessing under realistic settings. We use the parameter values below in the Markov chain model to reflect real-world scenarios 
and provide insights into \proj{}'s resilience.

\begin{itemize}[leftmargin=15pt, itemsep=0pt, parsep=0pt, topsep=2pt, partopsep=2pt]
    \item \( p = 100 \): Max num of malicious flow entries in a switch.
    \item \(m=16\): Attack to match \texttt{ether-type} (16-bit).
    \item \( n = 2^m - 276 = 65,260 \): Num of available \texttt{ether-type}s.
    \item \( q = 3 \): Requires 3 consistent MorphType discoveries.
    \item \( l = 5\): Sending discovery every 5 seconds.
\end{itemize}

The resulting probabilities are:
\begin{equation}\label{prob_in}
P_{\text{in}} = \frac{p}{n} = \frac{100}{65,260} \approx 0.0015323322
\end{equation}
\begin{equation}\label{prob_out}
P_{\text{out}} = 1 - P_{\text{in}} = 0.9984676678
\end{equation}
\inlinedsection{Attack and Recovery Time}
Due to $q=3$ and (\ref{law_total_exp}), we have:
\begin{equation}
E_0 = 1 + P_{\text{in}}E_1 + P_{\text{out}}E_0
\end{equation}
\begin{equation}
E_1 = 1 + P_{\text{in}}E_2 + P_{\text{out}}E_0
\end{equation}
\begin{equation}
E_2 = 1 + P_{\text{in}}E_3 + P_{\text{out}}E_0 \end{equation}
\begin{equation}
E_3 = 0
\end{equation}

By recursively solving these equations, the expected number of rounds to achieve 
3-hit streak is given by:
\begin{equation}
E_0 = \frac{1 + P_{\text{in}} + P_{\text{in}}^2}{1 - P_{\text{out}} - P_{\text{in}}P_{\text{out}} - P_{\text{in}}^2P_{\text{out}}} = %\frac{1+P_\text{in}+P_\text{in}^2}{P_\text{in}^3} = 
\frac{1}{P_\text{in}^3} + \frac{1}{P_\text{in}^2}+\frac{1}{P_\text{in}}
\end{equation}
Similarly, the expected number of rounds to achieve a 3-miss streak is given by:
\begin{equation}
M_0 = \frac{1 + P_{\text{out}} + P_{\text{out}}^2}{1 - P_{\text{in}} - P_{\text{out}}P_{\text{in}} - P_{\text{out}}^2P_{\text{in}}} = %\frac{1+P_\text{out}+P_\text{out}^2}{P_\text{out}^3} =
\frac{1}{P_\text{out}^3} + \frac{1}{P_\text{out}^2}+\frac{1}{P_\text{out}}
\end{equation}

By substituting (\ref{prob_in}) and (\ref{prob_out}) into the analysis, we have:
\begin{itemize}[leftmargin=15pt, itemsep=0pt, parsep=0pt, topsep=2pt, partopsep=2pt]
    \item \textbf{Expected attack time} \(\approx 278\) million rounds (\textasciitilde 44 years if 5 seconds per round).
    \item \textbf{Expected recovery time} \(\approx 3\) rounds.
\end{itemize}

These bounds upper-bound the rate at which MorphType random guessing can produce sustained fabrications; the end-to-end guarantee combines this rate with CamoType verification.

\subsection{CamoType Stealth Guarantee}
\label{sec:camoguarantee}

While Section~\ref{sec:markov} bounds MorphType resilience to random guessing, an advanced adversary may still fabricate links by
intercepting all unknown packets at a compromised switch or by selectively manipulating one probe layer. CamoType verification, triggered on every link change event, closes these gaps. 
An adaptive adversary aware of \proj{} may attempt to selectively manipulate individual probe types. Any such manipulation that fabricates a link, however, produces a reported link change and thus triggers CamoType verification, reducing every fabrication strategy to the dilemma below.
The following proposition formalizes the guarantee.

\begin{proposition}[CamoType Stealth Dilemma]
\label{prop:camo}
Let $\mathcal{A}$ be any adversary satisfying the Stealthy Adversary Assumption (\S\ref{sec:threat_model}). If any probe (DecoyType or MorphType) reports a link change differing from the last-confirmed topology (Algorithm~\ref{alg:verfi_detect}), then at CamoType verification exactly one of the following holds:
\begin{enumerate}[leftmargin=15pt, itemsep=0pt, parsep=0pt, topsep=2pt, partopsep=2pt]
  \item[(i)] $\mathcal{A}$ does not tamper with the CamoType probe.
  CamoType reveals the true link, and any probe disagreeing with it is
  flagged as manipulated per Algorithm~\ref{alg:verfi_detect} (lines~14--27).
  \item[(ii)] $\mathcal{A}$ tampers with the CamoType probe, manipulating
  normal data-plane traffic and exposing $\mathcal{A}$.
\end{enumerate}
Hence no $\mathcal{A}$ can both fabricate links and remain undetected.
\end{proposition}
\begin{proof}[Proof]
No adversary—malicious host, switch, application, or controller peer—can single out the CamoType: its identity is never revealed (§III), and CamoType is constructed to resemble ordinary ARP traffic and its identity is not disclosed outside the trusted controller. Under our adversary-knowledge and stealth assumptions, an attacker cannot reliably target CamoType without also risking interference with legitimate ARP traffic. Two cases arise. (i) 
$\mathcal{A}$ leaves the probe intact: CamoType returns the true link, and any disagreeing DecoyType or MorphType probe is flagged (Alg.~\ref{alg:verfi_detect}, lines 14–27). (ii) 
$\mathcal{A}$ redirects the probe: resembling ordinary ARP/IP traffic under our adversary model, doing so—whether a switch rewriting packets or an application or controller peer installing flow entries—also redirects normal traffic, violating constraint 1 and exposing 
$\mathcal{A}$. The cases are exhaustive, so no 
$\mathcal{A}$ can fabricate a link undetected.
\end{proof}
\section{System Evaluation}

We implement \proj{} on the open-source SDN controller OpenDaylight (Calcium release). We evaluate its effectiveness against known types of topology poisoning attacks in Section~\ref{sec:security_eval}, the fine-grained convergence latency of discovery in Section~\ref{sec:latency}, and the performance overhead in Section~\ref{sec:perf_eval}.

We deploy \proj{} on Ubuntu 22.04.4 with a 2.4 GHz 8-core Intel Core i9 processor and 32 GB of memory. CPU utilization follows Linux process accounting, where
100\% corresponds to one fully utilized core; thus, the 8-core
system has a total capacity of 800\%.
\proj{} is integrated as a submodule of OpenDaylight’s l2switch project~\cite{l2switch_code} with the \texttt{loopremover} modified to support reactive forwarding and a CamoType-specific ARP filter added to the \kw{arphandler} and \kw{hosttracker} to maintain compatibility with the host-tracker. The OpenFlow network is emulated using Mininet v2.3.0~\cite{mininet230} on another VM with the same system, CPU, and 16 GB of memory.

By default, \proj{} sends one DecoyType and one MorphType probe every 5 seconds during normal operation. The link discovered by MorphType ages in 40 seconds. A MorphType-reported change triggers the three-round MorphType verification, while a remaining link-change event triggers a single CamoType ARP probe for final verification.

To further obfuscate discovery behavior, 
MorphType frame sizes approximate the characteristic bimodal size distribution observed in real network traffic~\cite{john2007analysis}: 50\% are small frames (64--118 bytes), 40\% are large frames near the Ethernet MTU (1418--1518 bytes), and the remaining 10\% are distributed across the intermediate range. Each CamoType probe is additionally issued after a randomized delay of 0 to 2 seconds to hinder timing-based classification.
Administrators retain full control over the choice and scheduling of discovery packets, enabling secure topology discovery without exposing discovery-related information to any network component. 

\subsection{Effectiveness Against Topology Poisoning}\label{sec:security_eval}
\begin{figure}[t]
\centering{\includegraphics[width=0.38\textwidth]{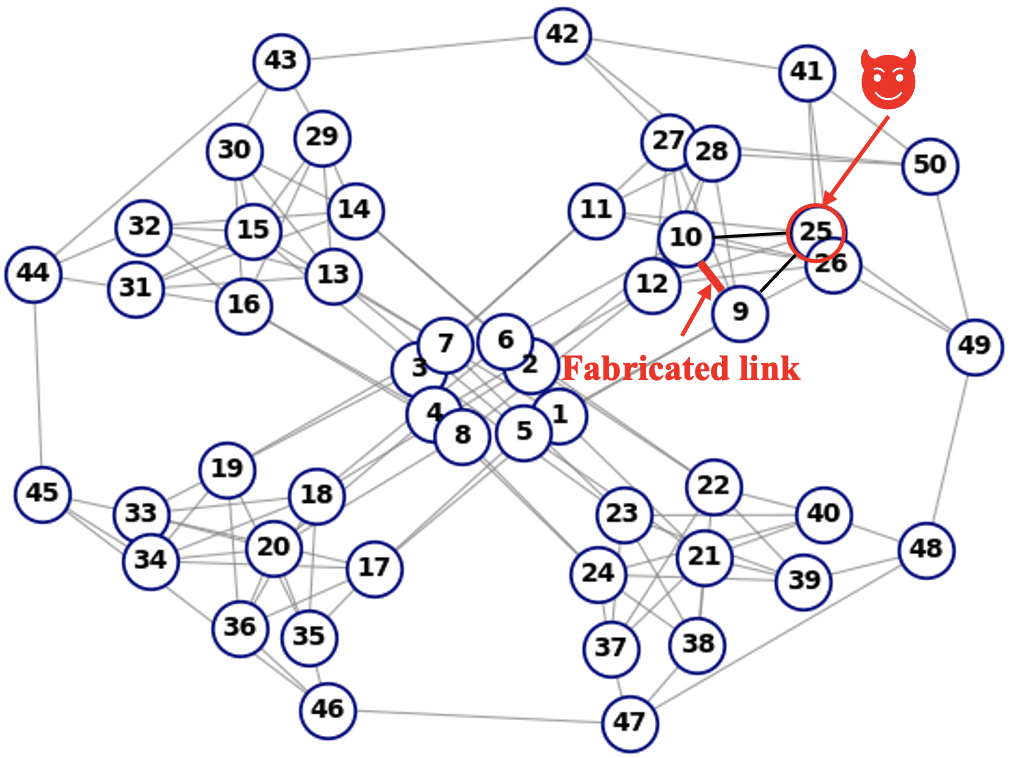}}
\vspace{-8pt}
\caption{\proj{} defeats relay attack}
\label{fig:50nodes_topo}
\vspace{-8pt}
\end{figure}

We evaluate \proj{} against all four types of topology poisoning attacks
defined in §\ref{sec:motivation}.%: spoof, replay, relay, and flow entry-induced poisoning. 
We use a 5-switch ring to verify the detection outcome; the complete per-attack results are shown in Appendix~\ref{appendix:attack_eval}. \proj{} prevents all four identifier-based attacks: DecoyType may be manipulated, while MorphType avoids identifier-specific forwarding
rules and CamoType verifies the resulting link change.

%Figure~\ref{fig:50nodes_topo} illustrates a relay attack on a \update{50-switch fat-tree data center topology} with 252 inter-switch links, in which Switch 25 forwards LLDP packets across its ports, fabricating a spurious link between Switch 9 and Switch 10. MorphType discovery is immune to this manipulation, and CamoType verification confirms the legitimate links around the affected ports. 

Figure~\ref{fig:50nodes_topo} illustrates a relay attack on a
50-switch fat-tree data-center topology with 252 inter-switch links. Malicious Switch~25 relays DecoyType LLDP packets,
causing the controller to infer the fabricated link between Switches~9 and~10. MorphType reports the genuine topology,
and the resulting disagreement triggers CamoType verification, which confirms the genuine neighbor and rejects the fabricated
link.

\subsection{Fine-Grained Convergence Latency}
\label{sec:latency}

We evaluate how quickly \proj{}'s accepted topology reflects four types of change: link down, link up, malicious relay, and 50\% packet loss. All experiments use the 252-link topology in
Figure~\ref{fig:50nodes_topo}. All events are launched from the \proj{} VM to the mininet VM via \textit{ssh}, and timestamps are taken on the controller's clock to avoid cross-VM synchronization issues.

We record 100 link-down/up pairs, 100 malicious-relay attacks, and
100 events of 50\% bidirectional packet loss (Table~\ref{tab:latency}). We restart the controller for each case and randomly select the affected links. Events are spaced 17.25\,s apart, deliberately offset from the 5\,s periodic discovery interval. The detailed latency distribution of each case is in Appendix~\ref{appendix:latency}.

\begin{table}[t]
\centering
\caption{Convergence latency by stage}
\label{tab:latency}
\vspace{-5pt}
\footnotesize
\setlength{\tabcolsep}{4pt}
\begin{tabular}{llrrrrr}
\hline
event & stage & \# events & median & IQR & p90 & max \\
      &       &           & (s)    & (s) & (s) & (s)  \\
\hline
link down & detect & 100 & 0.01 & 0.01--0.01 & 0.01 & 5.19 \\
link down & verify & 100 & 4.52 & 4.10--4.95 & 5.40 & 5.54 \\
link down & \textbf{converge} & 100 & \textbf{4.53} & 4.15--5.06 & 5.41 & 8.95 \\
\hline
link up & detect & 100 & 0.30 & 0.27--0.30 & 0.31 & 0.67 \\
link up & verify & 100 & 1.05 & 0.49--1.57 & 1.82 & 2.06 \\
link up & \textbf{converge} & 100 & \textbf{1.37} & 0.82--1.83 & 2.18 & 2.49 \\
\hline
mal-relay & detect & 100 & 2.95 & 1.52--4.12 & 4.70 & 5.18 \\
mal-relay & verify & 100 & 1.22 & 0.78--1.69 & 1.91 & 2.04 \\
mal-relay & \textbf{converge} & 100 & \textbf{4.06} & 2.85--5.16 & 5.96 & 7.21 \\
\hline
pkt loss & detect & 97* & 14.46 & 9.36--22.48 & 28.70 & 41.22 \\
pkt loss & verify & \multicolumn{5}{l}{sent to verify, but the probe never returned} \\
pkt loss & \textbf{converge} & 97* & \textbf{14.46} & 9.36--22.48 & 28.70 & 41.22 \\
\hline
\multicolumn{7}{l}{\footnotesize *Loss still allowed enough probes to refresh the link in 3 of 100 events.}\\
\end{tabular}
\vspace{-5pt}
\end{table}

Overall, \proj{} converges quickly for both legitimate changes and attacks. Link-down events are detected almost immediately through asynchronous OpenFlow \kw{port-status} notifications (0.01\,s median), allowing the controller to observe the failure without waiting for the next periodic discovery probe, but require 4.53\,s to converge because confirming an absent link requires waiting for unanswered probes. Link appearance is faster (1.37\,s median), as event-triggered probes return and allow verification to complete without silence timeouts. Malicious relays are detected in 4.06\,s median: It is first exposed when a periodic DecoyType probe reports a fabricated link, so its detection latency includes the wait until that probe occurs.

%The trigger is a link change notification returned by a Decoy LLDP packet, which takes longer than a prompt port status update. CamoType returns from the genuine neighbor and directly contradicts the fabricated link. 

Under 50\% bidirectional packet loss, 97 of 100 impaired links were retired within the 45\,s window, with a median convergence time of 14.46\,s. Because partial loss generates no port-status event and successful probes can still refresh the link, retirement is driven by periodic aging rather than change-triggered verification. CamoType is triggered only after the link is removed and therefore does not affect this convergence time. More generally, returned probes complete quickly, whereas missing replies incur configured timeout delays.

\subsection{Performance Overhead Analysis}
\label{sec:perf_eval}

We evaluate \proj{}'s controller overhead from three perspectives:
resource usage under attacks, the additional
packet-to-\textit{link-src} mapping, and scalability with network size.

\inlinedsection{Overhead Under Attacks}
We first measure \proj{}'s steady-state overhead relative to the baseline controller. On the 252-link topology, \proj{} uses 46.6\% mean CPU versus 33.6\% for the baseline, while post-GC live heap is 199 MB versus 193 MB, with substantially overlapping distributions. We then examine whether attack handling introduces additional transient overhead beyond this steady-state cost. As shown in Appendix~\ref{appendix:overhead},
the injected attack and recovery events produce no visible
CPU or retained-heap spikes. Thus, \proj{} incurs additional
steady-state discovery overhead, but processing an attack does
not measurably increase resource usage beyond that baseline.

%On the 252-link topology, \proj{} uses 46.6\% mean CPU versus 33.6\% for baseline (same controller basis but without the \proj{} module), while post-GC live heap is 199\,MB versus 193\,MB with substantially overlapping distributions. More importantly, the injected attack and recovery events cause no visible increase in either CPU or retained memory, indicating that \proj{}'s attack handling introduces negligible additional resource cost beyond its normal discovery overhead. Detailed measurements are provided in Appendix~\ref{appendix:overhead}.
%\input{overhead_subsection}

\inlinedsection{Scalability Across Network Sizes}
\begin{table}[t]
\caption{CPU \& mapping entries overhead across topo sizes}
\label{tab:scalability}
\footnotesize
\setlength{\tabcolsep}{4pt}
\begin{tabular}{ll rrrr rrrr}
\hline
(Node,  & \multirow{2}{*}{Ctrl} & \multicolumn{4}{c}{CPU (\%)} & \multicolumn{4}{c}{\# of mapping entries} \\
%\cmidrule(lr){3-6}\cmidrule(lr){7-10}
Link) & & min & max & \textbf{mean} & med. & min & max & mean & \textbf{med.} \\
\hline
(40, & base & 7.0 & 101.0 & \textbf{30.5} & 24.0 & -- & -- & -- & -- \\
%\cmidrule(lr){2-6}\cmidrule(lr){7-10}
282)  & cham & 11.0 & 118.0 & \textbf{36.0} & 29.0 & 40 & 80 & 63.0 & \textbf{60} \\
\hline
(50, & base & 3.0 & 125.7 & \textbf{36.2} & 36.0 & -- & -- & -- & -- \\
%\cmidrule(lr){2-6}\cmidrule(lr){7-10}
252) & cham & 5.0 & 124.0 & \textbf{38.8} & 32.0 & 50 & 100 & 71.1 & \textbf{74} \\
\hline
(80, & base & 13.0 & 103.0 & \textbf{50.1} & 48.5 & -- & -- & -- & -- \\
%\cmidrule(lr){2-6}\cmidrule(lr){7-10}
314) & cham & 8.0 & 230.0 & \textbf{57.2} & 51.0 & 80 & 160 & 121.6 & \textbf{120} \\
\hline
(120, & base & 39.0 & 137.0 & \textbf{73.5} & 70.0 & -- & -- & -- & -- \\
%\cmidrule(lr){2-6}\cmidrule(lr){7-10}
816) & cham & 2.0 & 461.0 & \textbf{96.5} & 73.2 & 60 & 240 & 161.8 & \textbf{150} \\
\hline
\end{tabular}
\vspace{-8pt}
\end{table}

We evaluate controller CPU and \proj{}'s mapping-table state across four fat-tree topologies ranging from $(N,L)=(40,282)$ to $(120,816)$.
Each experiment starts from a clean controller and starts monitoring once stabilized; CPU is sampled every
${\sim}1.2$\,s for 300\,s, while mapping-table size is recorded from the
controller approximately every 30\,s. As Table~\ref{tab:scalability} shows, the mapping table grows with topology size but remains bounded: its mean size increases from 63 entries at 40 switches to 162 at 120 switches, with a maximum of 240 entries on the 816-link topology. This behavior is consistent with its design: entries represent only outstanding probes and are removed when probes return or expire, so state scales with concurrently active discovery rather than controller uptime. Together, these results show that \proj{} introduces modest steady-state overhead while larger fabrics mainly increase burst processing and in-flight mapping state.

CPU also increases with topology size for both controllers. Relative to the baseline, \proj{} adds only 2.6--7.1 percentage points of mean CPU through 80 switches, while the 120-switch topology increases the mean by 23.0 points. The corresponding median increase at 120 switches is only 3.2 points, indicating that the larger mean is driven primarily by transient CPU bursts rather than sustained load. Such startup bursts can be mitigated by lengthening the MorphType discovery, which is demonstrated in the trade-off evaluation.

\begin{figure}[t]
\centering
\includegraphics[width=0.85\columnwidth]{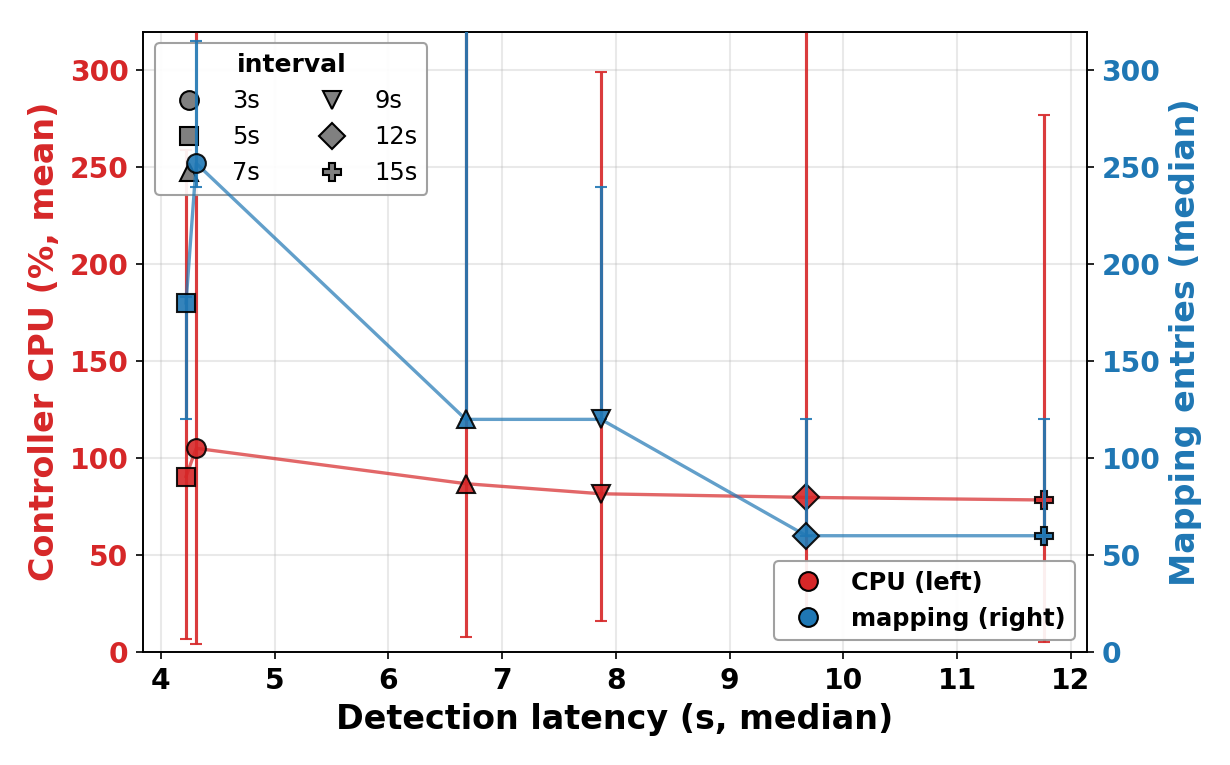}
\vspace{-10pt}
\caption{Trade-off between latency, CPU, and mapping entries}
\vspace{-15pt}
\label{fig:tradeoff}
\end{figure}

\inlinedsection{Trade-off between Detection and Overhead}
On the 120-switch topology, we vary the shared DecoyType/MorphType discovery interval from 3 to 15\,s while retaining the randomized CamoType delay of $[0,2]$\,s (Figure~\ref{fig:tradeoff} and Table~\ref{tab:tradeoff} in Appendix~\ref{appendix:tradeoff}). For each interval, we evaluate 30 relay attacks. Median detection latency increases from \textbf{4.3\,s} at a 3\,s interval to \textbf{11.8\,s} at 15\,s, while mean CPU utilization decreases from \textbf{105\%} to \textbf{78\%}. The median mapping-table size decreases from \textbf{252} entries at 3\,s to \textbf{60} entries at 12--15\,s, with a substantial reduction already achieved at 5--7\,s. These results expose a clear trade-off: shorter discovery intervals improve responsiveness but increase probe-processing load and the number of concurrent mappings, whereas longer intervals reduce overhead at the cost of slower attack detection. In our experiments, intervals of 5--7\,s provide a practical middle ground. 

Additional discussion is provided in Appendix~\ref{appendix:discussion}.

\section{Conclusion}
\noindent
\proj{} is a moving-target defense against topology poisoning attacks while remaining compatible with SDN workflows. Our analysis and evaluation demonstrate resilience against identifier-based and advanced relay attacks with practical convergence and a tunable security--performance trade-off.

%\proj{} combines dynamic randomization with decoy and camouflage mechanisms to prevent and detect identifier-based topology poisoning attacks while remaining compatible with existing SDN workflows.  MorphType provides probabilistic resilience through multi-round verification, while CamoType addresses stronger adversaries through stealth-based verification. Our implementation and evaluation demonstrate that \proj{} effectively detects both known identifier-based attacks and advanced relay attacks while providing practical convergence latency and a tunable security--performance trade-off.

%-------------------------------------------------------------------------------
% optional clearing of the page

\bibliographystyle{plainurl}
\bibliography{ref}
\appendix
\section{SDN Topology Discovery Implementations}~\label{appendix:ofdp_impl}
There is no standardized SDN topology discovery protocol. Although OFDP is the de-facto approach, open-source controllers differ in three key aspects:

\begin{itemize}[leftmargin=15pt, itemsep=0pt, parsep=0pt, topsep=0pt, partopsep=0pt]
\item the LLDP \kw{ether-header};
\item the payload encoding the \textit{link-src};
\item the flow entries used to return LLDP packets.
\end{itemize}

\begin{table}[t]
\small
\caption{LLDP Packet Fields in SDN Controllers~\cite{nehra2019tilak}}
\label{tab:LLDP_implementations}
\footnotesize
\setlength{\tabcolsep}{2.25pt}
\centering 
\begin{tabular}{llllll}
\toprule
\textbf{Controller} & \textbf{Ether-Dst} & \textbf{Ether-Src}  & \textbf{Ether-} & \textbf{Payload} & \textbf{Security}\\
& &  &\textbf{Type} & \textbf{\{Src Sw,} & \\
& & & & \textbf{ Port ID\}} & \\
\midrule
Ryu~\cite{ryu_code} & 01:80:c2: & Src Port& 0x88cc & Chassis ID,  & None\\
 & 00:00:0e& Mac Addr &  & Port ID & \\
\midrule
Pox~\cite{pox} & 01:23:20: & Src Port & 0x88cc & Chassis ID, & None\\
& 00:00:01& MAC Addr & & Port ID & \\
\midrule
Floodlight & 01:80:c2: & Src Port & 0x88cc & Opt.<Sw ID>,  & Static Hash\\
~\cite{floodlight} &  00:00:0e & MAC Addr & & Port ID & \\
\midrule
ODL~\cite{opendaylight_code} & 01:23:00: & Src Port & 0x88cc & Opt.<Sw ID, & Partial\\
 &  00:00:01 & MAC Addr & & Port ID> & Static Hash \\
\midrule
ONOS~\cite{onos_code} & a5:23:05: & Fingerprint & 0x88cc & Chassis ID, & Dynamic  \\
 & 00:00:01 &  &  & Port ID & Secret \\
\bottomrule
\end{tabular}
\vspace{-5pt}
\end{table}

Table~\ref{tab:LLDP_implementations} summarizes LLDP configurations across five controllers, showing variations in \kw{ether-dst}, \kw{ether-src}, and payload fields. Despite these differences, all rely on a fixed LLDP \kw{ether-type}—the \textbf{root cause} exploited by existing topology poisoning attacks. Floodlight, OpenDaylight (ODL), and ONOS embed secrets in LLDP packets to authenticate them, but remain vulnerable to spoofed \kw{packet-in} messages from compromised switches and to relay-based or flow-entry–induced attacks (Table~\ref{tab:link_fab_contermeasure}), which manipulate forwarding rather than packet contents. Floodlight’s static cryptographic hash is also vulnerable to replay attacks.

Table~\ref{tab:flow_entries4discovery} shows two main types of flow entries for discovery:
\begin{itemize}[leftmargin=15pt, itemsep=0pt, parsep=0pt, topsep=0pt, partopsep=0pt]
\item table-miss entries, which forward all unmatched packets to the controller;
\item LLDP-specific entries, which match on the LLDP \kw{ether-type} and forward packets to the controller.
\end{itemize}

Floodlight and ODL rely on table-miss entries, while ONOS, Ryu, and Pox install LLDP-specific entries. However, the discovery-related flow entries in Pox, Floodlight, ODL, and ONOS can be overridden by malicious high-priority rules from Marionette. Ryu assigns its LLDP-specific entries the highest priority, but attackers can lower default priorities or insert competing high-priority entries to attack. Matching behavior among equal-priority entries is switch-dependent~\cite{openflow}, complicating security in multi-vendor environments.

Overall, none of the five implementations provides comprehensive protection against topology poisoning attacks.
\begin{table}[t]
%% increase table row spacing, adjust to taste
%\renewcommand{\arraystretch}{1.3}
% if using array.sty, it might be a good idea to tweak the value of
% \extrarowheight as needed to properly center the text within the cells
\small
\caption{Flow Entries Assisting SDN Link Discovery~\cite{nehra2019tilak}}
\label{tab:flow_entries4discovery}
%\vspace{-0.5\baselineskip}
\footnotesize
\setlength{\tabcolsep}{2.25pt}
\centering 
%\begin{tabular}[htbp]{p{1.3cm}p{2cm}p{1cm}p{3cm}}
\begin{tabular}[htbp]{llll}
\toprule
\textbf{Controller} & \textbf{Match} & \textbf{Priority} &\textbf{Action}\\
\midrule
Ryu~\cite{ryu} & \kw{ether-type}: 0x88cc  & 65535 & To-controller\\
& \kw{ether-dst}: 01:80:c2:00:00:0e&& \\
\midrule
Pox~\cite{pox} & \kw{ether-type}: 0x88cc  & 65000 & To-controller\\
&\kw{ether-dst}: 01:23:20:00:00:01 && \\
\midrule
Floodlight~\cite{floodlight} & Table-miss  & 0  & To-controller\\
\midrule
ODL~\cite{opendaylight} & Table-miss  & 0  & To-controller \\
\midrule
ONOS~\cite{onos_disc} &  \kw{ether-type}: 0x88cc & 40000 & To-controller\\
\toprule
\end{tabular}
\vspace{-10pt}
\end{table}

\section{Case Study of LLDP in SDN Controllers}\label{appendix:lldp_detail}
We analyze the behavior of trending open-source SDN controllers.
As detailed in Section \ref{sec:motivation},  a controller sends discovery packets to switches using \kw{packet-out} messages, and these packets are returned from neighboring switches to the controller via \kw{packet-in} messages. Figure \ref{fig:messages_breakdown} decomposes the messages involved in the discovery process, including implementation examples from ONOS and OpenDaylight.

\begin{figure}[t]
\centering
\subfigure[\kw{Packet-out} Message Breakdown]{\label{fig:packet-out} \includegraphics[width=0.45\textwidth]{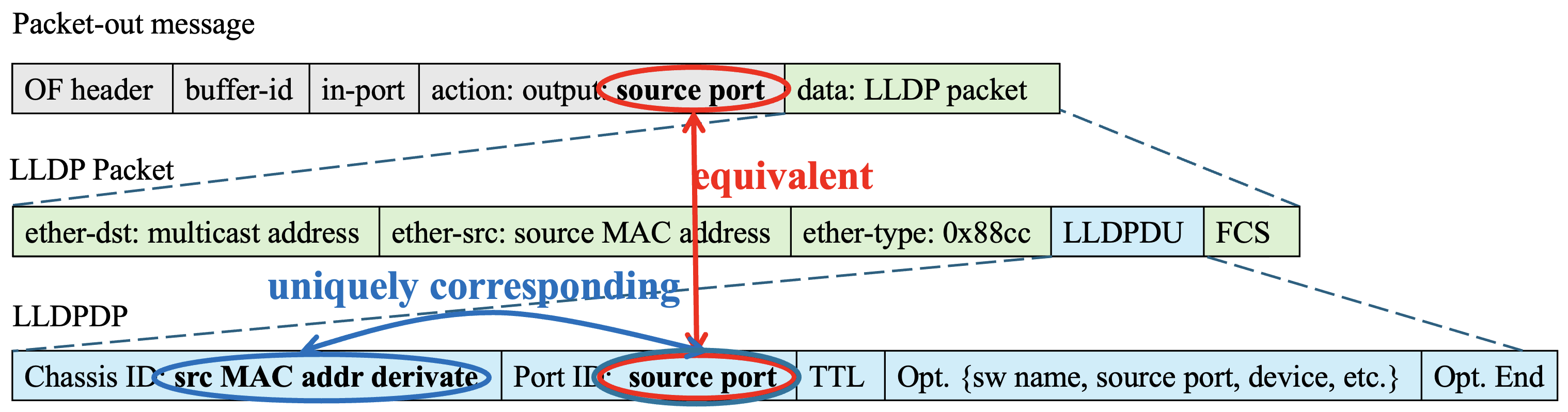}}
\subfigure[ONOS \kw{Packet-in} Message Breakdown]{\label{fig:onos_packet-in} \includegraphics[width=0.45\textwidth]{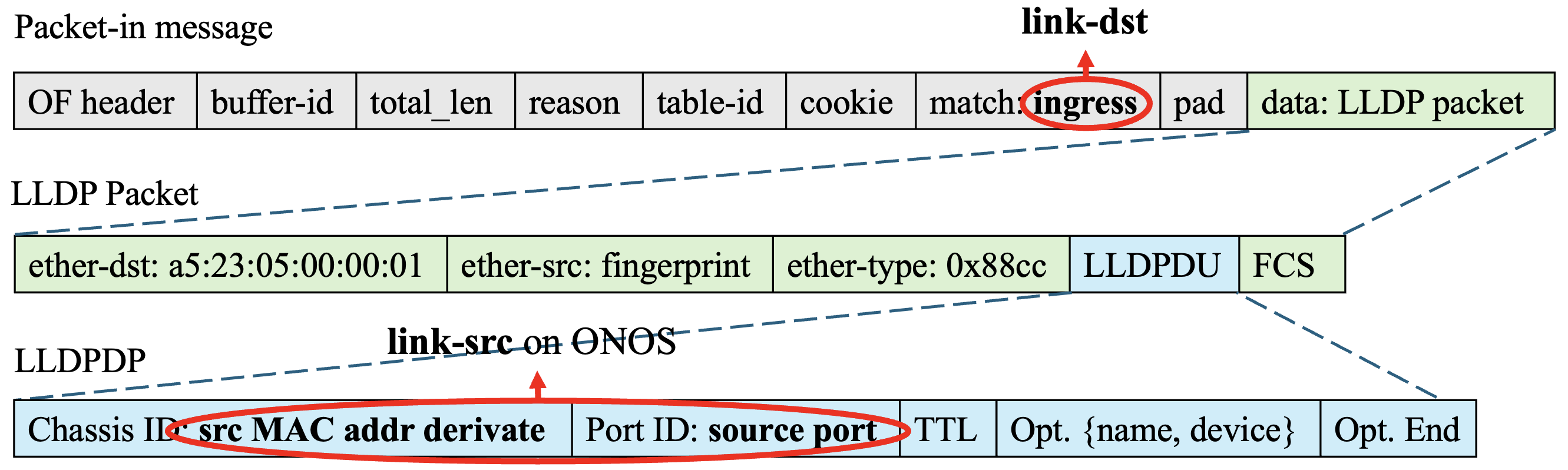}}
\subfigure[ODL \kw{Packet-in} Message Breakdown]{\label{fig:odl_packet-in} \includegraphics[width=0.45\textwidth]{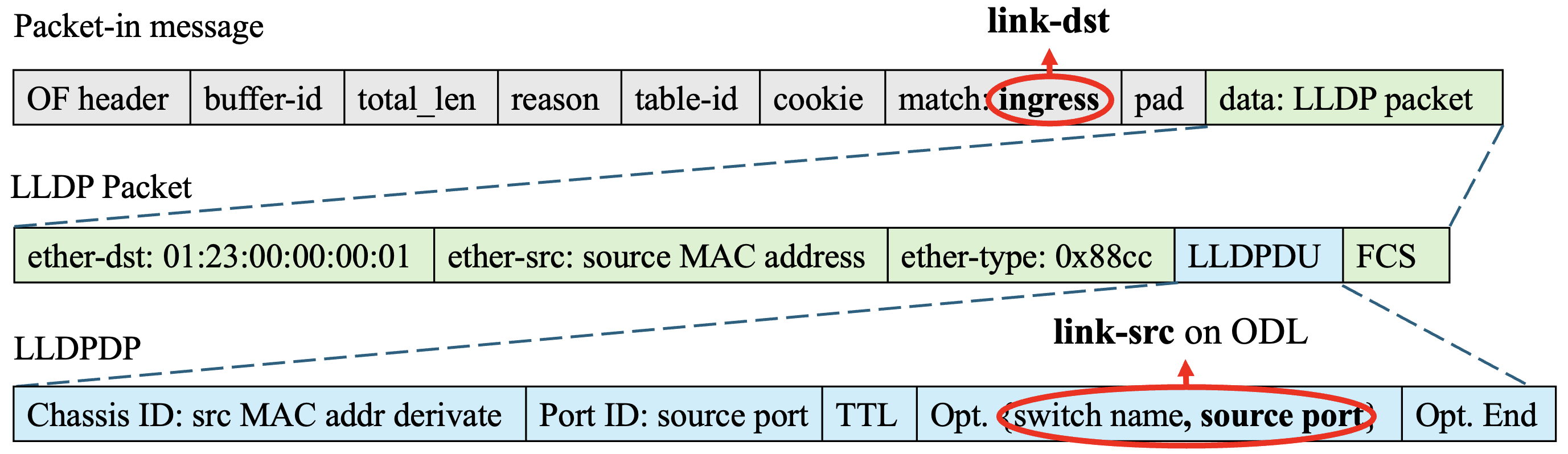}}
\subfigure[LLDP Packet Redundencies]{\label{fig:key_observation} \includegraphics[width=0.45\textwidth]{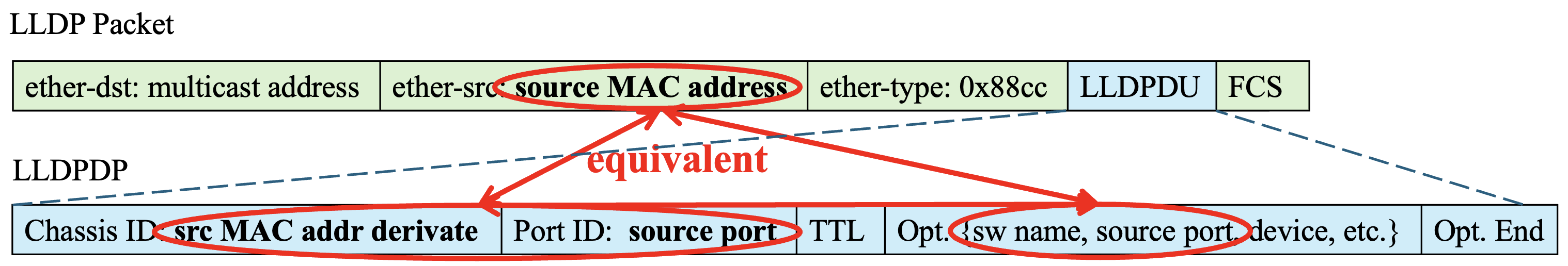}}
\caption{Field-Level Breakdown of Discovery Process Messages}
\label{fig:messages_breakdown}
\end{figure}

The key insight from Figure \ref{fig:packet-out} is that the LLDP packet is embedded as \kw{data} within the \kw{packet-out} message sent by the controller to the switch. Upon receiving this \kw{packet-out} message, the switch forwards the \kw{data} based on the instructions in the \kw{action} field. Crucially, the output port specified in the \kw{action} must correspond to the \kw{Port ID} associated with the LLDP packet to ensure accuracy. Each \kw{Chassis ID} uniquely corresponds to its \kw{Port ID} because it is a MAC address derivative of the source port. Notably, the MAC address of the source port also determines the value of \kw{ether-src}.

When a neighboring switch receives the LLDP packet, it forwards it back to the controller due to the table-miss or LLDP flow entry, as discussed in Section \ref{sec:motivation}. Figures \ref{fig:onos_packet-in} and \ref{fig:odl_packet-in} illustrate the differing implementations of the \kw{packet-in} message for ONOS and OpenDaylight controllers, respectively. In both cases, the controllers determine the \textit{link-dst} based on the ingress port of the \kw{packet-in} message. The distinction lies in determining the \textit{link-src} (\textit{link-src}), which requires parsing the LLDP packets. In the ONOS implementation, the \textit{link-src} is inferred from the mandatory \kw{Chassis ID} and \kw{Port ID} fields, while optional fields contain additional "name" and "device" information. Conversely, the OpenDaylight implementation derives \textit{link-src} information from the optional fields, which store "switch name" and "source port" details initially recorded by the controller.

From this example, we can see that there are various ways to store \kw{link-src} information. However, to the best of our knowledge, all open-source controllers, including both ONOS and OpenDaylight, overlook that \kw{ether-src} alone is sufficient to retrieve \textit{link-src} information and fail to utilize it for this purpose. That is because:
(1) \kw{ether-src} (the source switch port MAC address) is equivalent to the \kw{Chassis ID} (derivation from the
source port’s MAC address); (2) the SDN controller maintains a one-to-one mapping between port MAC addresses and switch Port IDs in its data store. Therefore, \textbf{Observation 1} is justified. Given the controller’s knowledge of all switches, any packet with the \kw{ether-src} set to a source switch port can be used for topology discovery.

\section{Non-Conflicting Randomization}\label{appendix:reason}

Selecting randomized MAC addresses that are not present in the SDN network eliminates the need for dynamic data plane monitoring and ensures that normal flow entries do not inadvertently match
discovery packets.
\begin{itemize} [leftmargin=15pt, itemsep=0pt, parsep=0pt, topsep=2pt, partopsep=2pt]
    \item Hosts and switches, which have MAC addresses, change infrequently in a non-wireless network\footnote{When devices are replaced or reconfigured (months to years)}; so, their corresponding flow entries, if any, also tend to remain stable.
    \item Flow entries matching individual Ethernet source or destination addresses are rare\footnote{SDN rules typically match on IP and TCP/UDP not MAC addresses.}. %and wildcarding MAC addresses for matching is also less common\footnote{Wildcarding MACs adds overhead, offers limited policy value, and is inconsistently supported in hardware.}.
    \item Flow entries that match only on \kw{ether-type} typically forward control packets to the controller (e.g.,  ARP and LLDP), aligning with \proj{}’s goal of ensuring discovery packets reach the controller.
\end{itemize}

\section{Markov Chain Model}~\label{markov_appx}

\begin{figure}[t]
\centering
\subfigure[Markov Chain of $q$-hit streak]{\label{fig:markov_coincidence}
\includegraphics[width=0.23\textwidth]
{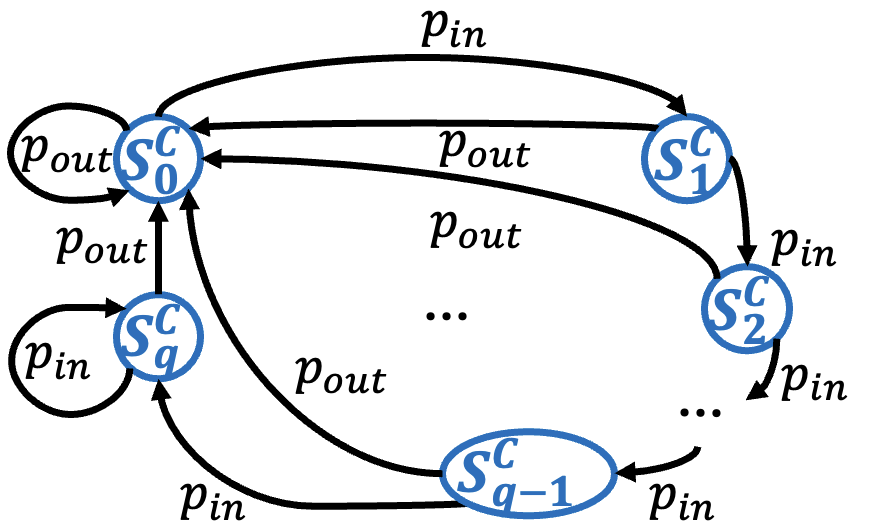}}
\hfill
%\vspace{-5pt}
\subfigure[Transition Matrix of $q$-hit streak]{\label{fig:transition_coincidence} \includegraphics[width=0.23\textwidth]{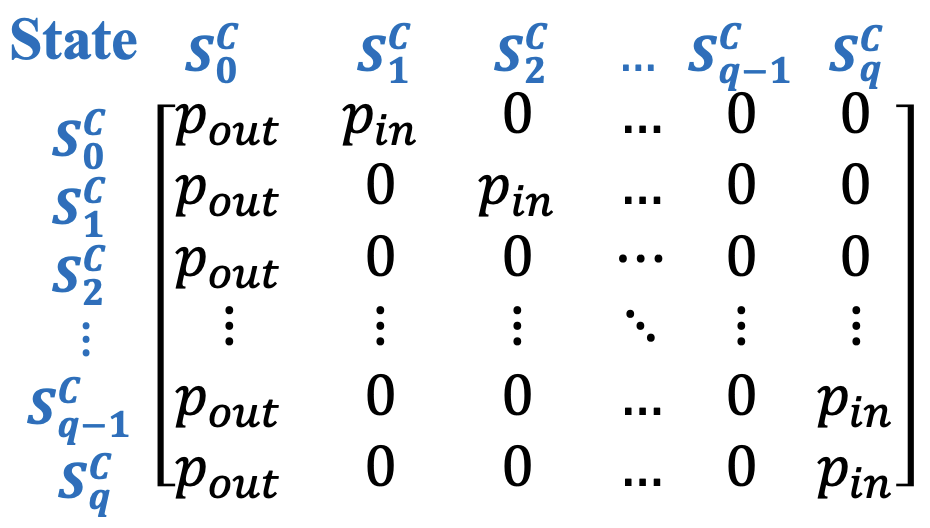}}
\caption{Markov Chain Model of $q$-hit streak}
\label{fig:markov_model_coincidence}
%\vspace{-10pt}
\end{figure}

In Figure~\ref{fig:markov_model_coincidence}, we depict the Markov Chain Model for the MorphType resilience evaluation.

\section{Effectiveness Against Topology Poisoning}\label{appendix:attack_eval}

We intentionally use a small 5-switch ring topology in Figure 3 (hosts on SW1/SW4 omitted for clarity) as a proof-of-concept to demonstrate detection correctness, which is topology-independent since detection operates per port. 
We simulate all four types (i.e., spoof, replay, relay, and flow entry-induced) of topology poisoning attacks on a 5-switch ring topology in Mininet. Using Scapy, we craft spoofed packets to launch spoof-based attacks and capture authentic LLDP packets to simulate replay-based attacks. We install malicious flow entries to simulate relay-based and flow entry-induced poisoning. All four attacks ultimately manipulate LLDP packet handling --- either the packet content (spoof, replay) or its forwarding (relay, flow entry-induced) --- so each generation procedure is short and self-contained. %\trent{Is there some easily explained guide for how the attacks are generated - or hanmany instances of attacks?  Not required, but if easy to say would be nice to have.}\mm{added}

\begin{figure}[t]
\centering{\includegraphics[width=0.49\textwidth]{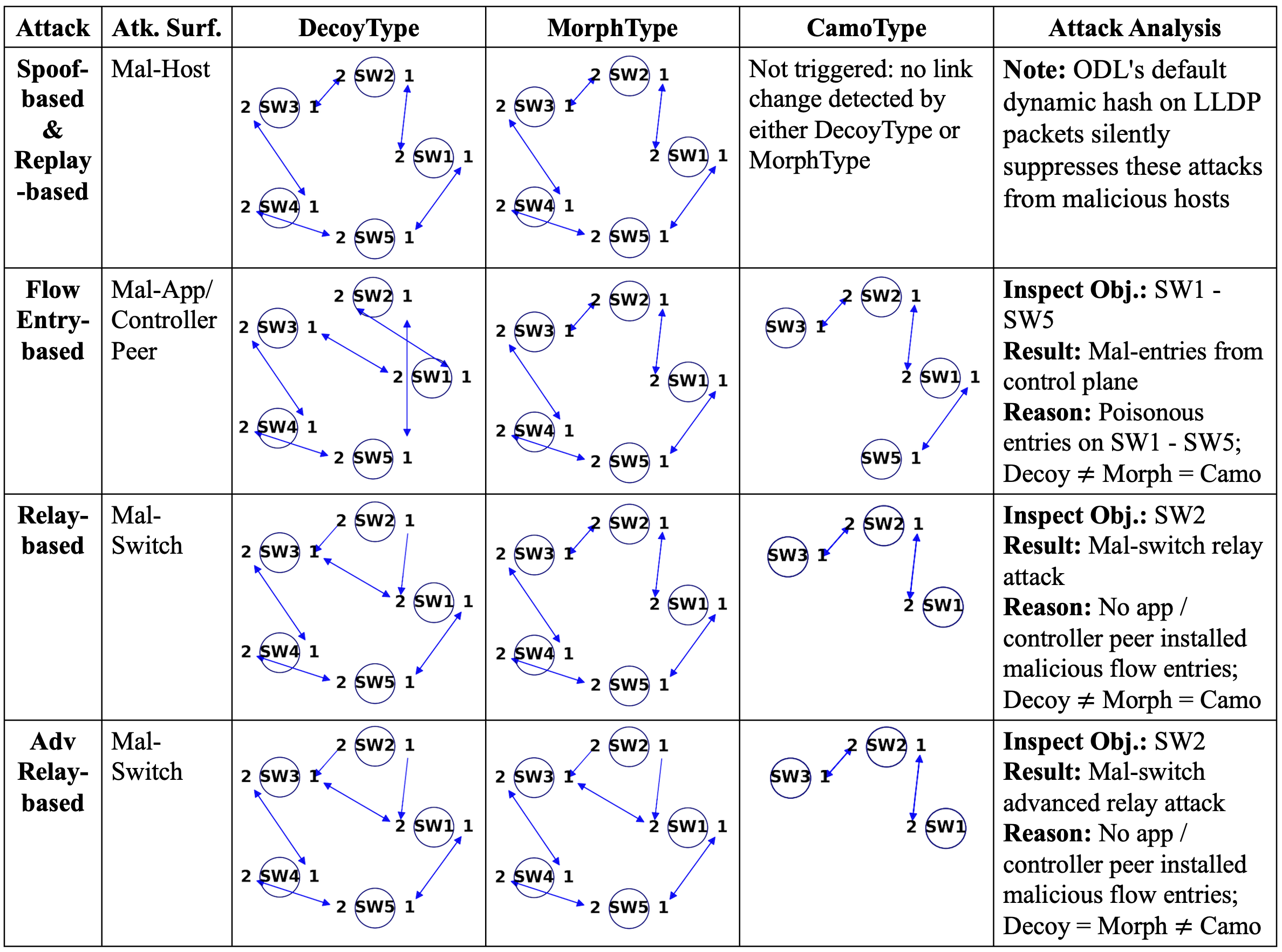}}
%\vspace{-20pt}
\caption{\proj{} with topology poisoning attacks}
\label{fig:defeat_all_topo_pois}
\vspace{-10pt}
\end{figure}

\proj{} prevents all four types of identifier-based topology poisoning attacks (\S\ref{sec:motivation}), as illustrated in Figure~\ref{fig:defeat_all_topo_pois}. Whenever DecoyType or the verified MorphType result reports a link change,
CamoType verification is triggered, and the resolution depends on which probes agree with CamoType (Algorithm 1, lines 14–27). Spoof- and replay-based attacks from malicious hosts are silently suppressed by ODL's dynamic LLDP hash, which \proj{} inherits without modification — DecoyType discovery in Figure~\ref{fig:defeat_all_topo_pois} reflects this protection. Flow entry-induced attacks from malicious applications or controller peers, and relay-based attacks from malicious switches, both mislead DecoyType discovery by manipulating LLDP forwarding via poisonous flow entries or in-port relays. In the above cases, MorphType reveals the true topology because its randomized headers ensure only table-miss rules match, blocking malicious flow entries and bypassing static-identifier exploitation. 
%\trent{Why do we need Camo then?  Sounds like you didn't try all the attacks.} \mm{thx for the catch, I changed it to all above attacks and the next paragraph discuss the advanced relay attack.}
CamoType serves as the final verification on every link change event — whether DecoyType and MorphType agree on a new link or disagree. The reported change is accepted only when ARP probes — which blend with legitimate data-plane traffic — confirm it.

\proj{} also defends against a more advanced threat in which a malicious switch intercepts all unknown packets arriving on a specific \kw{in-port}—an attack we refer to as an advanced relay attack, also shown in Figure~\ref{fig:defeat_all_topo_pois}. Malicious Switch 2 blindly relays both LLDP packets and all unknown packets arriving on Port 1 to Port 2, and vice versa, regardless of flow entries. This behavior causes both LLDP and MorphType discovery to be misled, fabricating a bidirectional link between Switch 1 and Switch 3. However, CamoType discovery remains unaffected, as ARP packets resemble normal data-plane traffic that attackers are unlikely to manipulate to maintain stealth.

\section{Ethernet Types of \proj{}}\label{appendix:ether-type}
Table~\ref{tab:ether-type} summarizes the Ethernet types suitable for \proj{} to support secure topology discovery.
\begin{table*}[t]
\small
\caption{Ethernet Types Chosen by \proj{}}
\label{tab:ether-type}
%\vspace{-0.5\baselineskip}
\footnotesize
\setlength{\tabcolsep}{2.25pt}
\centering 
\begin{tabular}[htbp]{p{1.2cm}p{1cm}p{1.2cm}p{14cm}}
%\begin{tabular}[htbp]{llll}
\toprule
\textbf{DiscType}  & \textbf{PktType} & \textbf{EthType} & \textbf{Description}\\
\midrule
Decoy & LLDP & 0x88cc  & Advertising identity and capability information to neighbor switches for discovery purposes \\

\cmidrule{2-4}
& BDDP & 0x8999  & Leveraging broadcast frames to advertise link information to discover hybrid links by collaborating with LLDP packets\\

\midrule
Morph & Rand & Unassigned & Randomly select an EtherType value, verify it remains unassigned in the IEEE Registration Authority~\cite{ieee_registration}, and ensure no conflict with other internal standards \\

\midrule
Camo & ARP & 0x0806 & Resolves IP-to-MAC address mappings; suitable for camouflaged verification and naturally supported by default flow entries.\\

\cmidrule{2-4}
& IPv4/ IPv6 & 0x0800/ 0x86DD & Ubiquitous data-plane types suitable for camouflage, but require careful construction to meet header standards and avoid flow rule conflicts.\\
\toprule
\end{tabular}
%\vspace{-10pt}
\end{table*}

\section{Trade-off Table}~\label{appendix:tradeoff}
Table~\ref{tab:tradeoff} shows the details of CPU, number of mapping entries, and detection latency measurements with varied Decoy and MorphType intervals.
\begin{table*}[t]
\centering
\caption{Trade-off between detection latency, CPU utilization, and mapping-table size}
\label{tab:tradeoff}
\footnotesize
\setlength{\tabcolsep}{4pt}
\begin{tabular}{r rrrr rrrr rrrr}
\hline
interval & \multicolumn{4}{c}{latency (s)} & \multicolumn{4}{c}{CPU (\%)} & \multicolumn{4}{c}{mapping entries} \\
\cline{2-5}\cline{6-9}\cline{10-13}
(s) & min & max & mean & med & min & max & mean & med & min & max & mean & med \\
\hline
3 & 1.6 & 9.9 & 4.5 & 4.3 & 4 & 411 & 105 & 97 & 240 & 315 & 260 & 252 \\
5 & 1.2 & 11.2 & 5.4 & 4.2 & 7 & 259 & 90 & 79 & 120 & 183 & 171 & 180 \\
7 & 2.6 & 26.8 & 8.1 & 6.7 & 8 & 483 & 87 & 80 & 120 & 676 & 160 & 120 \\
9 & 2.3 & 18.6 & 8.8 & 7.9 & 16 & 299 & 82 & 73 & 120 & 240 & 141 & 120 \\
12 & 2.8 & 20.5 & 9.8 & 9.7 & 13 & 351 & 80 & 70 & 60 & 120 & 70 & 60 \\
15 & 3.8 & 19.2 & 11.5 & 11.8 & 5 & 277 & 78 & 72 & 60 & 120 & 71 & 60 \\
\hline
\end{tabular}
\end{table*}

\section{Mapping-Table State}
\begin{figure}[t]
\centering
\includegraphics[width=\columnwidth]{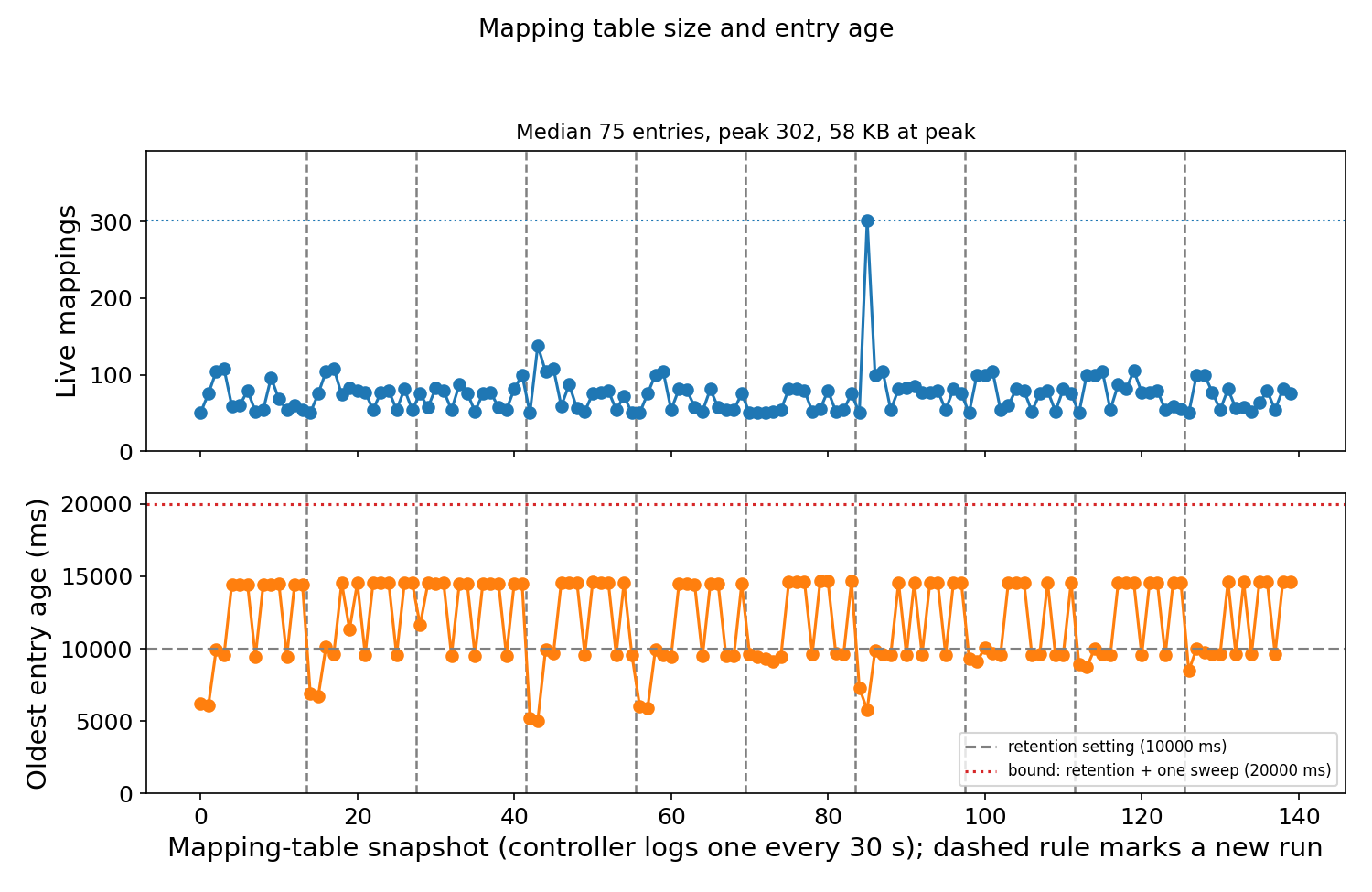}
\vspace{-20pt}
\caption{Mapping-Table Overhead Under Attacks}
\label{fig:entries_overhead}
\end{figure}

Figure~\ref{fig:entries_overhead} confirms this behavior over 90 minutes, six controller restarts, and 216 injected link-down and link-up changes. On the topology in Figure~\ref{fig:50nodes_topo}, the table contains a median of 79 entries and peaks at 83, requiring only 16 KB. The oldest-entry age also remains bounded
by the retention interval, confirming that entries are continuously retired rather than accumulated.

\section{Performance overhead under attacks}\label{appendix:overhead}
\begin{figure}[t]
\centering
\includegraphics[width=\columnwidth]{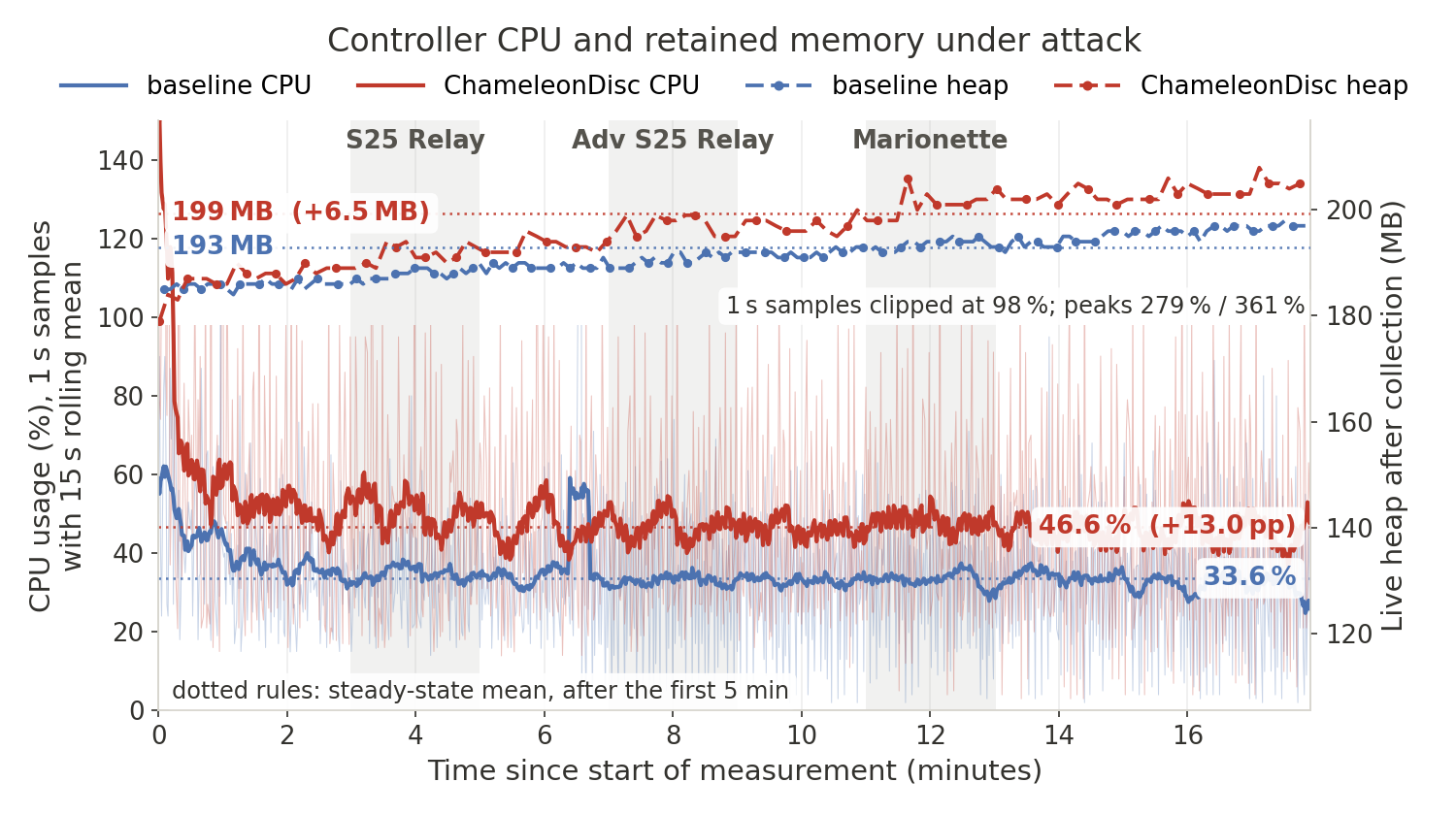}
\vspace{-20pt}
\caption{CPU and Heap Overhead Under Attacks}
\vspace{-15pt}
\label{fig:overhead-attack}
\end{figure}
We measure \proj{} against the \emph{baseline} implementation on the topology of Figure~\ref{fig:50nodes_topo}.
Each configuration runs for 17 minutes, with three attack/recovery pairs injected at fixed intervals (Figure~\ref{fig:overhead-attack}). CPU is sampled roughly once per second. Because JVM resident-set size does not reliably reflect live application memory, we measure post-GC live heap from the JVM garbage-collection log and validate it with a forced full collection at the end of each run.

As Figure~\ref{fig:overhead-attack} shows, attacks and recoveries do not visibly perturb CPU or retained memory. On average, \proj{} uses 46.6\% CPU versus 33.6\% for the baseline, an increase of 13.0 percentage points, although \proj{} approximately doubles
periodic discovery probes by transmitting both DecoyType and MorphType. The live-heap distributions substantially
overlap: \proj{} settles around 199\,MB versus 193\,MB for the baseline (179--208\,MB versus 184--198\,MB). 
The packet-to-\textit{link-src} mapping is \proj{}'s principal additional persistent state. Each entry is created when a probe is sent and removed when it returns or expires. Consequently, the table is bounded by probes in flight rather than controller uptime. The mapping table in this case has around 50-83 entries, which is around 16kB. The next section evaluates the scalability via CPU and mapping table size with various topologies.

Although ARP-based CamoType discovery relies on a static packet type and is theoretically targetable, manipulating such common data-plane traffic carries a high risk of exposure—contradicting the assumptions of our stealth-focused threat model. Manipulating IP-based CamoType discovery packets presents even greater risk since IP is the predominant packet type in network traffic, further discouraging adversarial interference. However, composing IP-based CamoType packets introduces considerable overhead because it requires dynamically collecting and analysing flow entries on each switch~\cite{alimohammadifar2018stealthy}. We therefore choose ARP-based CamoType.  

\section{Latency Distribution}\label{appendix:latency}
Figure~\ref{fig:latency_linkupdown} shows the latency distributions for link-up and link-down events, Figure~\ref{fig:latency_attack} shows the distribution under relay attacks, and Figure~\ref{fig:latency_networkfault} shows the distribution when monitored links experience 50\% packet loss.

\begin{figure}[t]
\centering
\includegraphics[width=\columnwidth]{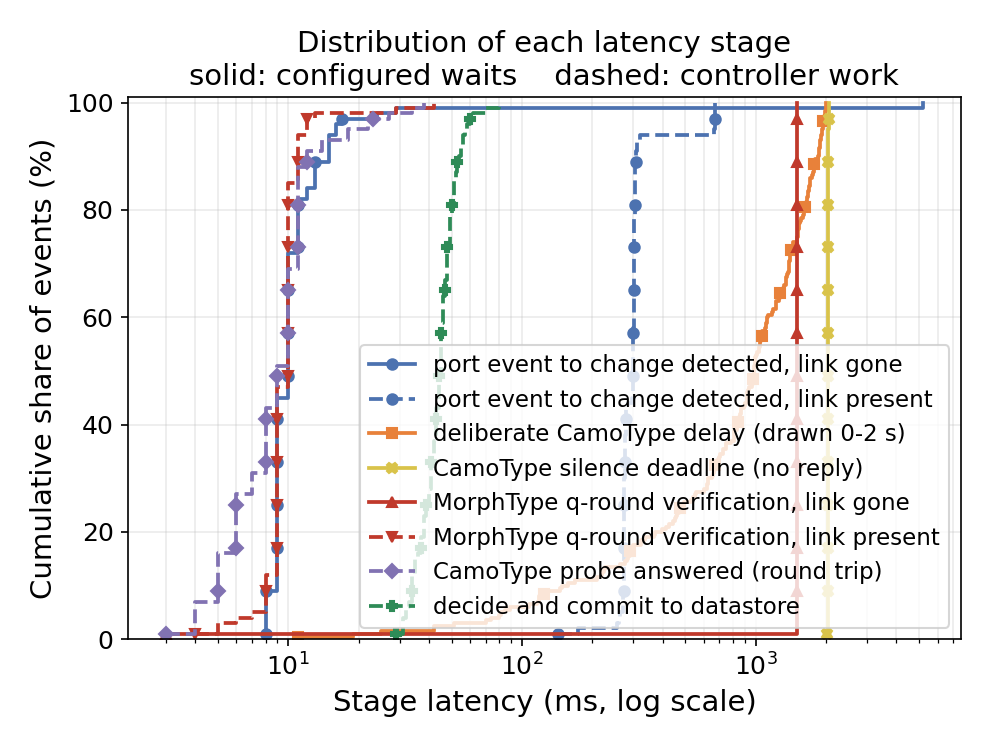}
\vspace{-20pt}
\caption{Latency Distribution While Link Up and Down}

\label{fig:latency_linkupdown}
\end{figure}

\begin{figure}[t]
\centering
\includegraphics[width=\columnwidth]{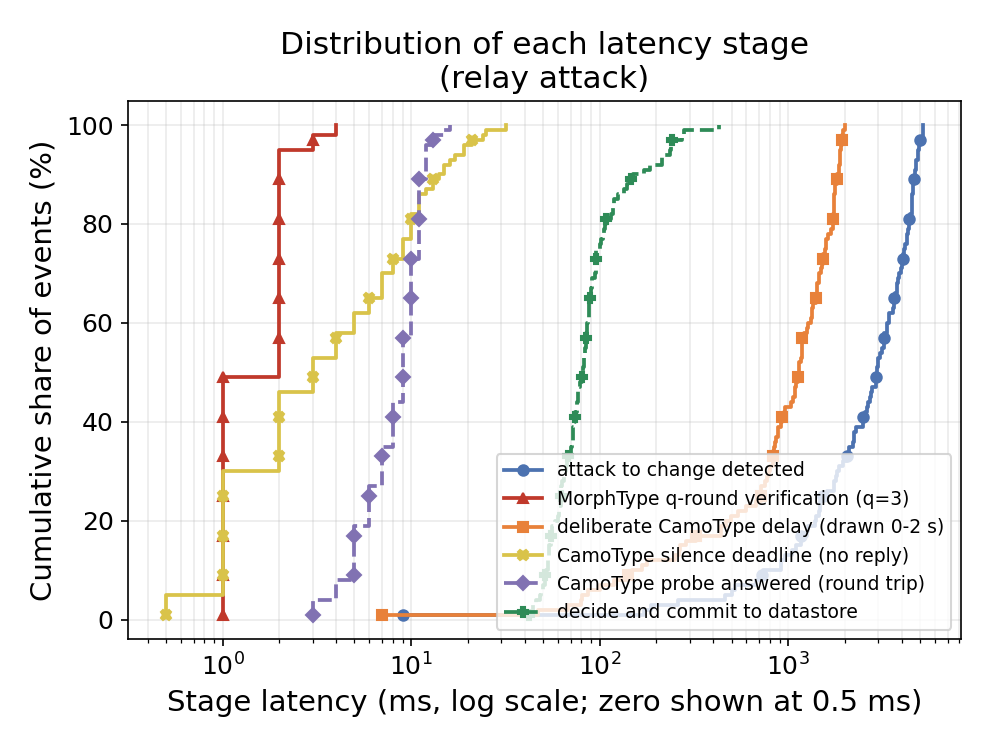}
\vspace{-20pt}
\caption{Latency Distribution While Relay Attack}

\label{fig:latency_attack}
\end{figure}

\begin{figure}[t]
\centering
\includegraphics[width=\columnwidth]{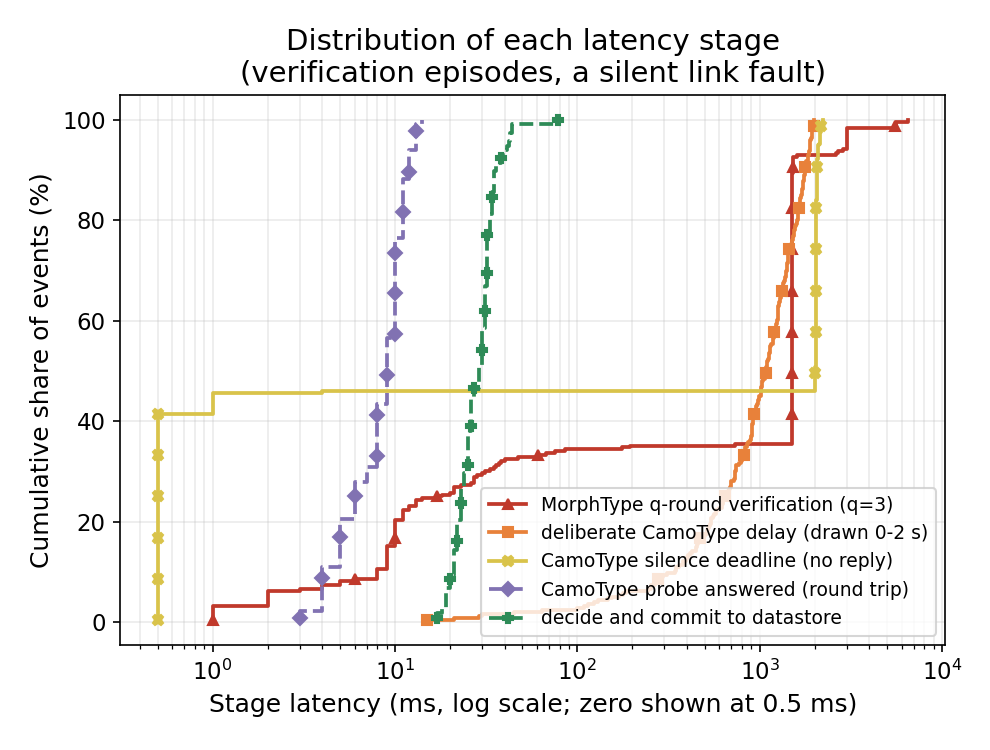}
\vspace{-20pt}
\caption{Latency Distribution While 50\% Packet Loss}

\label{fig:latency_networkfault}
\end{figure}

\section{Discussion}

\inlinedsection{Limitations} \proj{} targets common reactive forwarding setups that avoid coarse-grained flow rules matching \kw{in-port}. In purely proactive networks — rare in practice — MorphType becomes ineffective, requiring administrators to manually configure \kw{to-controller} entries and rely on CamoType (e.g., ARP, IP).%, which adds monitoring overhead and side-channel exposure risk. Our threat model also assumes stealthy adversaries that avoid manipulating normal data-plane traffic; aggressive learning-based adversaries that monitor long-term traffic and accept exposure to manipulate normal flows fall outside this scope and remain future work.

\inlinedsection{Compatibility with Hybrid SDNs} \proj{}'s core design --- randomized discovery packets validated centrally --- extends naturally to hybrid SDNs and other southbound interfaces (e.g., P4Runtime, BGP-based control). Adaptation primarily requires rethinking discovery initiation and \kw{packet-in} processing, which we leave to future work.

%\inlinedsection{Randomizing Discovery Pattern} An attacker may analyze packet intervals to identify discovery packets. \proj{} randomizes MorphType packet sizes (bimodal distribution) and CamoType verification delays (0 - 2 s) to frustrate timing-based classification. Randomizing transmission intervals for MorphType and DecoyType further is future work, at the cost of variable detection latency.

\inlinedsection{Open Source}  
Once accepted, we will release \proj{} as open source to foster collaboration and further improvements from the security and networking community. %\proj{} introduces a lightweight, probabilistically sound approach to SDN link discovery that integrates proactive prevention with real-time detection—rooted in a practical threat model. We plan to extend CHAMELEONDISC to hybrid and large-scale SDN environments, evaluating its adaptability under realistic workloads.

\label{appendix:discussion}
\section{Ethical Considerations}
This work focuses solely on defense mechanisms.
All referenced vulnerabilities are previously disclosed with public CVEs; no new vulnerabilities or human-subject data are involved.

\section*{Open Science}
We provide an anonymized artifact containing the source code, experiment scripts and configurations, raw evaluation data, and analysis scripts used in this paper.
\url{https://anonymous.4open.science/r/ChameleonDisc-B484/}
\cleardoublepage
%%%%%%%%%%%%%%%%%%%%%%%%%%%%%%%%%%%%%%%%%%%%%%%%%%%%%%%%%%%%%%%%%%%%%%%%%%%%%%%%
\end{document}